\documentclass[11 pt,reqno]{amsart}

\usepackage{amsfonts,amssymb}
\usepackage{amsthm, amssymb, latexsym, amsmath}
\newtheorem{theorem}{Theorem}

\newtheorem{lemma}[theorem]{Lemma}

\newtheorem{remark}[theorem]{Remark}
\newtheorem{definition}[theorem]{Definition}

\numberwithin{equation}{section}

\newcommand{\nn}{\nonumber}
\newcommand{\C}{\mathbb{C}}
\newcommand{\R}{\mathbb{R}}

\newcommand{\cF}{\mathcal{F}}

\newcommand{\cQ}{\mathcal{Q}}

\newcommand{\lam}{{\lambda}}           

\newcommand{\Einfty}{E(\infty)}

\newcommand{\1}{{\bf 1}}

\newcommand{\DETAILS}[1]{}

\newcommand{\supp}{\operatorname{supp}}

\newcommand{\Ran }{\operatorname{Ran}}

\renewcommand{\Re}{\mathrm{Re}}

\newcommand{\pp}{\phi_1 \otimes \phi_2}
\newcommand{\ppy}{\phi_1\otimes\phi_{2,\di}}

\newcommand{\phipsi}[2]{\varphi_{n_{#1}}\otimes\psi_{m_{#2}}}
\newcommand{\Psinm}[2]{\Psi^{n_{#1}m_{#2}}}
\newcommand{\Od}{O_d}
\newcommand{\Odi}{\Od^\infty}

\newcommand{\groundst}{\zeta}
\newcommand{\di}{r}
\newcommand{\chix}{\chi}
\newcommand{\chiz}{\widetilde{\chi}}
\newcommand{\Ix}{I}
\newcommand{\Iz}{\widetilde{I}_\di}

\newcommand{\Izij}{\widetilde{I}_{ij,\di}}
\newcommand{\Rx}{R}
\newcommand{\Rrx}{\Rx}

\newcommand{\Rz}{\widetilde{R}}

\newcommand{\psix}{\psi}
\newcommand{\psiz}{\widetilde{\psi}}
\newcommand{\Hx}{H}
\newcommand{\Hxabot}{H^{a,\bot}}
\newcommand{\Hz}{\widetilde{H}}

\newcommand{\Px}{P}
\newcommand{\Pxbot}{\Px^\bot}
\newcommand{\Pz}{\widetilde{P}}
\newcommand{\Pzbot}{\Pz^\bot}
\newcommand{\gsi}{\zeta_1}
\newcommand{\gsii}{\zeta_2}
\newcommand{\YN}{Y}
\newcommand{\Imp}{G}
\newcommand{\subi}{\mathcal{S}_1}
\newcommand{\subj}{\mathcal{S}_2}
\newcommand{\subij}{\subi\times\subj}
\newcommand{\Chix}{X}

\newcommand{\fr}{f_\di}
\newcommand{\frz}{\widetilde{\fr}}
\newcommand{\Tspin}{\widetilde{T}_\pi}

\title[Van der Waals interaction between two atoms]{Differentiability of the van der Waals interaction between two atoms}

\author[I. Anapolitanos]{Ioannis Anapolitanos}
\address{Ioannis Anapolitanos, Dept.~of Math., Karlsruhe Institute of Technology, Karlsruhe, Germany} 
\email{ioannis.anapolitanos@kit.edu}

\author[M. Lewin]{Mathieu Lewin}
\address{Mathieu Lewin, CNRS \& CEREMADE, University Paris-Dauphine, PSL University, 75016 Paris, France} 
\email{mathieu.lewin@math.cnrs.fr}

\author[M. Roth]{Matthias Roth}
\address{Matthias Roth, Dept.~of Math., Karlsruhe Institute of Technology, Karlsruhe, Germany} 
\email{matthias.roth@kit.edu}

\begin{document}

\begin{abstract}
In this work we improve upon previous results on the expansion of the interaction energy of two atoms. On the one hand we prove the van der Waals-London's law, assuming that only one of the ground state eigenspaces of the atoms is irreducible in an appropriate sense. On the other hand we prove strict monotonicity of the interaction energy at large distances and, under more restrictive assumptions, we provide the leading order of its first two derivatives. The first derivative is interpreted as the force in Physics. Moreover, the estimates of the first two derivatives provide a rigorous proof of the monotonicity and concavity of the interaction energy at large distances. 
\end{abstract}

\maketitle

\section{Introduction and main result}

Atoms and molecules attract each other through van der Waals forces, which are much weaker than ionic or covalent bonds. These forces have been discovered by J.D. van der Waals~\cite{vdW1,vdW2}
when he was trying to formulate an equation of state of gases, compatible with experimental measurements. These forces are universal and play a fundamental role in quantum chemistry, physics, biology and material
sciences. For instance, their strength is one of the factors that determine the boiling temperature of liquids. 
They explain why diamond, consisting of carbon atoms connected with
covalent bonds only, is a much harder material than graphite, which
consists of layers of carbon atoms that attract each other through van der Waals forces~\cite{C}.

Our goal in this paper is to discuss the differentiability of the interaction energy of two atoms at dissociation and the long range behavior of its first two derivatives, justifying thereby the long range behavior of the van der Waals force. We work under the Born-Oppenheimer approximation, in which the two nuclei are classical particles and the electrons are quantum.

We begin with a mathematical formulation of the problem. We study the Hamiltonian
\begin{equation}\label{Hdec}
\Hx=\Hx(\di)= \Hx_0 +\Ix,\qquad 
\Hx_0 = \Hx_1+\Hx_{2,\di},
\end{equation}
where
\begin{equation*}
\label{def:H1H2}
\Hx_1=\sum_{i=1}^{N_1}\left(-\Delta_{x_i}-\frac{N_1}{|x_i|}\right) + \sum_{1\leq k < l \leq N_1}\frac{1}{|x_k - x_l|},
\end{equation*}
\begin{equation*}
\label{def:H2}
\Hx_{2,\di}=\sum_{j=N_1+1}^{N_1+N_2}\left(-\Delta_{x_j}-\frac{N_2}{|x_j-\di e_1|}\right) + \sum_{N_1+1\leq m < n \leq N_1+N_2}\frac{1}{|x_m - x_n|}
\end{equation*}
and
\begin{equation*}
\label{def:Ixij}
\Ix = \sum_{i=1}^{N_1}\sum_{j=N_1+1}^{N_1+N_2}\Ix_{ij}, \qquad
\Ix_{ij}=\frac{1}{\di}-\frac{1}{|x_i-\di e_1|}-\frac{1}{|x_j|}+\frac{1}{|x_i-x_j|}.
\end{equation*}
Here we assume that the first nucleus is at $0 \in \R^3$ and the second at $\di e_1,$ where  $e_1=(1,0,0)$ and $\di>0$ denotes the distance between the nuclei. The coordinates $x_1,\dots,x_{N_1+N_2} \in \R^3$,
are the locations of the electrons, where $N_1,N_2$ are the atomic numbers of the first and second atom, respectively. The operators $\Hx_1,\Hx_{2,\di}$ describe the individual atoms, and the multiplication operator $\Ix$ is the interaction between them. The notation $\Hx_{2,\di}$ emphasizes the fact that the nucleus of the second atom is located at a distance $\di$ from the origin. We work in units where the electron charge is $-1$, the electron mass is $0.5$ and Planck's constant is $\hbar=1$.

In order to take spin and the Fermi statistics into account, we introduce some orthogonal projections onto appropriate subspaces of
\begin{equation}\label{def:YN1N2}
\YN =L^2\bigg(\big(\mathbb{R}^3 \times \{0,...,q-1\}\big)^{N_1+N_2},\C\bigg)\simeq \bigotimes_{1}^{N_1+N_2}L^2\big(\mathbb{R}^3,\C^q\big).
\end{equation}
Here $q$ is the number of spin states, hence $q=2$ for electrons.
The corresponding norm is
$$\|\Phi\|_Y^2:= \sum_{s_j\in\{0,...,q-1\}}\int |\Phi(x_1,s_1,...,x_{N_1+N_2}, s_{N_1+N_2})|^2 dx_1 \cdots dx_{N_1+N_2}.$$
Sometimes we write $\int ds$ for the sum over spin states for shortness.
For any permutation
$\pi \in S_{N_1+N_2}$ we define
the unitary operator $T_{\pi}$ on $\YN$ which exchanges the variables of $\Psi$ as follows 
\begin{multline}\label{def:Tpi}
(T_{\pi} \Psi)(x_1,s_1,\dots,x_{N_1+N_2}, s_{N_1+N_2})\\=\Psi(x_{\pi^{-1}(1)},s_{\pi^{-1}(1)} \dots,
x_{\pi^{-1}(N_1+N_2)}, s_{\pi^{-1}(N_1+N_2)}),
\end{multline}
where $s_1,\dots,s_{N_1+N_2} \in \{0,...,q-1\}$ are the spin variables of the electrons.
Let $\subi=S_{N_1}\otimes\1 \subset S_{N_1+N_2}$ be the subgroup of $S_{N_1+N_2}$ consisting of the permutations of $\{1,...,N_1\}$, and $\subj=\1\otimes S_{N_2} \subset S_{N_1+N_2}$ be the subgroup of $S_{N_1+N_2}$ consisting of the permutations of $\{N_1+1,...,N_1+N_2\}$. We define, for $j\in\{1,2\}$,
$$\cQ_j=\frac{1}{N_j!}\sum_{\pi \in \mathcal{S}_j}(-1)^\pi T_\pi,\qquad \cQ=\frac{1}{(N_1+N_2)!}\sum_{\pi \in S_{N_1+N_2}}(-1)^\pi T_\pi.$$
Let us also introduce the operator 
$\Hx_{12}:= \Hx_1 +  \Hx_2$
where $\Hx_2=\Hx_{2,0}$. Consider the projected operators onto the appropriate symmetry subspaces
\begin{equation}\label{eq:symHam}
\Hx_{1}^a=\Hx_{1} \cQ_1,\qquad \Hx_{2}^a=\Hx_{2} \cQ_2,\qquad \Hx_{12}^a=\Hx_{12} \cQ_1 \cQ_2,\qquad
\Hx^a=\Hx \cQ.
\end{equation}
We call
 \begin{equation}\label{def:Er}
\boxed{E_j:=\min\sigma(H_j^a),\quad  E(r):=\min\sigma(H^a)}
  \end{equation}
the corresponding ground state energies. From the Zhislin and HVZ theorems~ \cite{Zh,vW,H,Sig,Lewin-11}, these are eigenvalues lying strictly below the essential spectrum. In particular we have $E_j<0$. Moreover $E(r)<0$ if $r$ is large enough, see \cite{LT}. Because of these properties, it is not important whether we consider the operators over the whole space $Y$ or over the appropriate anti-symmetric subspaces, since by definition they are equal to 0  outside of these spaces. In this paper we will always work in $Y$ for convenience. Note also that the bottom of the spectrum of $H_{12}^a$ is just $E_1+E_2$. 
  
For  $i \in \{1,2\}$ and $m \in \mathbb{Z}$ with $m \leq N_i$ and $N_i-m \leq N_1 + N_2$ we define
 \begin{equation*}\label{def:Hm}
 H_{i,m}:=\sum_{j =1}^{N_i-m} \left(- \Delta_{x_j}- \frac{N_i}{|x_j|}\right)+ \sum_{1 \leq j<k \leq N_i-m} \frac{1}{|x_j-x_k|},
 \end{equation*}
which is the Hamiltonian of the ion with atomic number $N_i$
 and total charge $m$. We define $E_{i,m}$ to be the ground state energy of $H_{i,m}^a=H_{i,m}\cQ$ in $Y$.
It has been proved in \cite[Thm.~1.1]{MoS} (see also \cite[Thm.~2]{Le}) that $E(r)$ converges at infinity and that the limit is given by
\begin{equation}
 E(\infty):=\lim_{r \rightarrow \infty} E(r)=\min_{-N_2\leq m \leq N_1} \Big\{E_{1,m} + E_{2,-m}\Big\}.
 \label{eq:value_limit}
\end{equation}
In other words, at the dissociation the two atoms are in their ground states but one still has to optimize over all the ways of distributing the $N_1+N_2$ electrons among them. 
The interaction energy of the system is defined by
 \begin{equation}\label{def:interaction'}
 \boxed{W(\di)=E(\di)- E(\infty). }
 \end{equation}
We will explain shortly that $W(\di) < 0$. This means that it costs energy to separate the atoms, hence the system must be bound at a finite distance $r$. Local minima of $W(\di)$ determine the equilibrium configurations of the diatomic molecule.  
  
We now discuss some known properties of the interaction energy $W(r)$. If the minimum on the right side of~\eqref{eq:value_limit} is attained at some $m\neq0$, then the two atoms have opposite charges in the dissociation. In this case they attract each other at infinity and this results in the upper bound
$$W(r)\leq -\frac{m^2}{r}+o\left(\frac1r\right)$$
(see, e.g.,~\cite{Le}). In particular $W(r)<0$ for $r$ large enough. 
In this paper we are interested in the case where the minimum on the right side of~\eqref{eq:value_limit} is solely attained at $m=0$, that is,
\begin{equation}\label{eq:neutralsplit}
 E_{1} + E_{2}= E_{1,0} + E_{2,0} <   \min_{\substack{-N_2\leq m \leq N_1\\ m\neq0}} \Big\{E_{1,m} + E_{2,-m}\Big\}.
 \end{equation}
It is indeed a famous conjecture that~\eqref{eq:neutralsplit} is satisfied for all $N_1,N_2$. See \cite{AS} and references therein for a discussion, including the case of several atoms. In this paper~\eqref{eq:neutralsplit} will be an assumption throughout. In particular we then have
$$W(r)=E(r)-E_1-E_2.$$
Under the assumption~\eqref{eq:neutralsplit}, the van der Waals-London's law for a system of two atoms asserts that there exists $\sigma>0$ such that
 \begin{equation}\label{vdWlaw}
\boxed{ W(r)=-\frac{\sigma}{r^6}+O\left(\frac{1}{r^7}\right)_{r\to\infty}.}
 \end{equation}
A heuristic explanation was given by London~\cite{Lo} in 1937.

One of the first rigorous results in this direction were given by  Morgan and Simon~\cite{MoS}, who proved, for two atoms without spin and with non-degenerate ground states, that $W(r)$ possesses an asymptotic series in powers of $1/r$. They also mentioned that the $6$th order term is to be interpreted as the van der Waals energy. In 1986,  Lieb and Thirring~\cite{LT} provided an upper bound of the form
$$W(r)\leq -\frac{c}{r^6}$$
for a positive constant $c>0$. The upper bound holds without any assumptions and it also applies to systems of several molecules. After these results,~\eqref{vdWlaw} was proven only relatively recently in \cite{AS} under a type of non-degeneracy assumption in the case of a system of several atoms. 

The goal of this article is to investigate the expansion of derivatives of $W(r)$. Our work extends the results in \cite{AS} in the case of two atoms in several directions.  
In the spinless non-degenerate case we provide estimates on the first two derivatives of the interaction energy from which monotonicity and concavity of the interaction energy follows for large $r$. 
Then we prove the van der Waals-London's law with spin, with the sole assumption that one of the ground state eigenspaces is irreducible in a sense to be discussed below. Moreover, with the help of the methods developed in the spinless case we prove the monotonicity of the interaction energy for large~$r$.

Throughout this work we assume that at least one of the ground state eigenspaces of the two individual atoms is \emph{irreducible}, a concept that we discuss now. The usual method is to introduce a group of spin transformations (that is, unitary operators acting on $L^2(\{0,...,q-1\})=\mathbb{C}^q$), and to require that the first eigenspace of $H^a$ is irreducible with respect to this group action. For simplicity we introduce here one group which works for every $N_1,N_2$. But the arguments below work for any other group representation of the spins, which acts non trivially on fermionic wavefunctions. For $\pi\in S_n$, we introduce the operator exchanging only the spin variables:
\begin{equation}\label{def:Tspin}
\Tspin\Psi(x_1,s_1,\dots,x_n,s_n) := \Psi(x_1,s_{\pi^{-1}(1)},\dots,x_n,s_{\pi^{-1}(n)}).
\end{equation}
Similarly, we may introduce the operator which shifts the spins by one unit
\begin{equation}\label{def:flip}
\mathcal{F}\Psi(x_1,s_1,\dots,x_n,s_n) := \Psi(x_1,s_1+1,\dots,x_n,s_n+1)
\end{equation}
where $s_j+1$ is always understood modulo $q$. For $q=2$, $\cF$ is just the operator which flips each spin to its opposite. Note that $\mathcal{F}$ commutes with $\Tspin$ for all $\pi$, hence these operators generate a finite group of order $qN!$. If $\Psi$ is a ground state of $H^a$ then so is $\mathcal{F}^p\Tspin\Psi$ for any $\pi\in S_N$ and any $p\in\{0,...,q-1\}$, since $H^a$ is does not depend on the spin variables.

\begin{definition}[Irreducibility of ground state eigenspace]\label{def:irreducible}
The ground state eigenspace $G_k$ of $\Hx_k^a$ is called \emph{irreducible} if $\{0\}$ and $G_k$ are its only invariant subspaces under the spin shift $\mathcal{F}$ and all the spin permutations $\widetilde{T}_\pi$ for all $\pi\in S_{N_j}$.
\end{definition}

Equivalently, the eigenspace is spanned by any ground state $\Psi$  and its spin transformations $\mathcal{F}^p\Tspin\Psi$ with $p\in\{0,...,q-1\}$ and $\pi\in S_N$. In particular, the ground state is unique up to spin appropriate rearrangements. In the case without spin, $q=1$, the irreducibility condition simply means that $\dim(G_k)=1$, that is, the eigenvalue is simple. With spin $1/2$ ($q=2$), the irreducibility assumption is verified in some natural physical situations which we discuss in Remarks~\ref{rem:nd1} and~\ref{rem:nd2} below. For $N_k=1$ the ground state eigenspace is always of dimension $q$, by Perron-Frobenius theory~\cite[Sec.~XIII.12]{RSIV}. The corresponding space is then always irreducible, due to the action of the shift $\mathcal{F}$. 

If $G_k$ is the ground state eigenspace of $\Hx_k^a$ then $\Hx_{12}^a$ has the  ground state eigenspace $G_1 \otimes G_2$.
We further define
\begin{equation*}
\Hx_{12}^{a,\bot}:=\Hx_{12}^a (1-\Px_{G_1 \otimes G_2}),
\end{equation*}
where $\Px_{G_1 \otimes G_2}$ denotes the orthogonal projection onto $G_1 \otimes G_2$. 
Since the ground state energy $E(\infty)=E_1+E_2$ is in the discrete spectrum of $H_{12}^a$ and since $E(\infty)<0$, there exists $c>0$ such that
$\Hx_{12}^{a,\bot}-E(\infty) \geq c$ on the whole space $\YN$.
As a consequence, the resolvent $R_{12}^{ \bot}:=(\Hx_{12}^{a,\bot}-E(\infty))^{-1}$ is well defined, bounded and positive on the entire space $\YN$.
 We now introduce 
 \begin{equation}\label{def:fij}
 f(z_1,\dots ,z_{N_1+N_2}) := \sum_{i=1}^{N_1}\sum_{j=N_1+1}^{N_1+N_2}f_{ij},
 \end{equation}
where $f_{ij} = z_i \cdot z_j -3 (z_i \cdot e_1)(z_j \cdot e_1) $  is the dipole-dipole energy. We finally consider the number
 \begin{equation}\label{sigmaijvdef}
\boxed{ \sigma:=\max_{\substack{\Psi \in G_1 \otimes G_2\\ \|\Psi\|=1}} \langle f \Psi, R_{12}^{\bot} f \Psi \rangle.}
\end{equation}
We recall that any ground state in $\zeta_j\in G_j$ is exponentially decaying,
 \begin{equation}\label{eq:expdec}
 	\|e^{d|x|}\partial^\alpha\groundst_j\|_{L^2} < \infty, \text{ for all } \alpha \text{ with, } |\alpha| \leq 2,
 \end{equation}
 see for example \cite{CT}.
 It follows that the  function $f \Psi$
 is in $L^2$ for any $\Psi \in G_1 \otimes G_2$ and thus $\sigma$ is finite.
 Using that $\langle f \Psi, R_{12}^{\bot} f \Psi \rangle=\|(R_{12}^{\bot})^\frac{1}{2} f \Psi\|^2$
one can prove that $\sigma>0$, see \cite[Section 3]{A}, and \cite[Section 2.2]{AL} for the more general case of two molecules.
The following theorem was proved in~\cite{AS}.

\begin{theorem}[van der Waals law in the irreducible case~{\cite{AS}}]\label{thm:old}
Assume \eqref{eq:neutralsplit}. For $j=1,2$, we also assume that the ground state eigenspace of $\Hx_{j}^a$ is irreducible (Definition~\ref{def:irreducible}). Then
there exist positive constants $D_1, D_2$ such that for all $\di \geq D_1$ we have
\begin{equation}\label{eq:old}
\bigg|W(\di)+\frac{\sigma}{\di^6}\bigg| \leq \frac{D_2}{\di^7}.
\end{equation}
\end{theorem}

\medskip

\begin{remark}[Lennard-Jones]·\rm 
The limit~\eqref{eq:old} justifies the replacement of $W(r)$ by the Lennard-Jones potential $cr^{-12}-dr^{-6}$ at infinity. The highly repulsive part $r^{-12}$ is itself purely empirical. The Lennard-Jones potential was initially introduced to infer an appropriate law of dependence of the viscosity of a gas on the temperature~\cite{J1} and to study the equation of state of gases~\cite{J2}. 
\end{remark}

Theorem \ref{thm:old} justifies the long-range part of $W$ but does not provide the monotonicity nor the concavity of $W(r)$ for large $r$. The inequality \eqref{eq:old} does not exclude the possibility that $W(r)$ oscillates for large $r$, and in particular it does not exclude existence of local minima for large $r$. Since $W'(r)$ is interpreted as the van der Waals \emph{force}, it is of course important to give its large-$r$ behavior as well.

Our results help in addressing these issues and give a more detailed description of the long range behavior of the interaction energy. The first theorem provides the leading order of the first two derivatives of $W(\di)$, in the case without spin $q=1$. In the second theorem spin is taken into account and the van der Waals-London's law is proven with the irreducibility assumption on only one of the two atoms. Moreover the strict monotonicity of $W(\di)$ is proven for $r\gg1$, but no exact expansion of the derivatives is obtained. 

\begin{theorem}[Spinless case]\label{thm:new}
Let $q=1$ (no spin). Assume \eqref{eq:neutralsplit} and that the ground state eigenspaces of $\Hx_{j}^a$, $j=1,2$ are both non-degenerate. 
Then there exist positive constants $D_1, D_2$ such that for all $\di \geq D_1$ we have 
\begin{equation}
\bigg| W(\di)+\frac{\sigma}{\di^6}\bigg| \leq {\frac{D_2}{\di^7}},\qquad \bigg| W'(\di)-\frac{6 \sigma}{\di^7}\bigg| \leq {\frac{D_2}{\di^8}},\qquad \bigg| W''(\di)+\frac{42 \sigma}{\di^8}\bigg| \leq {\frac{D_2}{\di^9}}.
 \label{eq:derexp}
\end{equation}
\end{theorem}

\begin{theorem}[Case with spin]\label{thm:new2}
Let $q\geq1$. Assume \eqref{eq:neutralsplit} and that one of the two ground state eigenspaces of $\Hx_{j}^a$ is irreducible. Then
there exist positive constants $D_1, D_2$ such that for all $\di \geq D_1$ we have
\begin{equation}\label{eq:new2}
\bigg|W(\di)+\frac{\sigma}{\di^6}\bigg| \leq \frac{D_2}{\di^7}.
\end{equation}
Moreover, $W(\di)$ is strictly increasing for $\di$ large enough.
\end{theorem}

\begin{remark}\label{rem:nd1}\rm 
If we do not take the fermionic statistics into account (which corresponds to $q=N$), the non-degeneracy assumption of the ground state eigenspaces is always satisfied, by Perron-Frobenius theory~\cite[Sec.~XIII.12]{RSIV}. Then a result similar to Theorem~\ref{thm:new} holds. 
\end{remark}

\begin{remark}\label{rem:nd2}\rm 
The condition \eqref{eq:neutralsplit} is known to hold when both atoms are hydrogens, that is, $N_1=N_2=1$~\cite{A,AS,FGRS}. In addition, the ground state eigenspace of the H$_2$ molecule is a singlet state for all $r>0$ when q=2, hence it is non-degenerate. We conclude that~\eqref{eq:derexp} holds for the H$_2$ molecule.

For the helium atom $N_j=2$ with spin ($q=2$), the ground state is a singlet state and hence it is therefore non-degenerate. We see that, under the assumption~\eqref{eq:neutralsplit}, the conclusions of Theorem \ref{thm:new2} hold for a hydrogen or helium atom interacting with any other atom. 
\end{remark}


\begin{remark}\rm  The irreducibility assumption and the neglected spin in part of our work are unsatisfactory. Unfortunately we have been unable to drop these conditions. The technical reason is that the Feshbach map, which we introduce and use below, is a matrix that is not necessarily a multiple of the identity if there are degeneracies. Its lowest eigenvalue is therefore not necessarily smooth, due to possible crossings. For this reason without the non-degeneracy assumption we were not even able to prove the differentiability of $W(r)$. Note that to prove the monotonicity of $W(r)$ in Theorem \ref{thm:new2} we proceed in a way that does not involve smoothness of $W(r)$, adapting an argument from~\cite{AL}.
\end{remark}

\begin{remark}\rm 
In the spinless case $q=1$, it is known that $W(r)$ is a real-analytic function of $r$ in the non-degenerate case~\cite{CoSe,H1}. Furthermore, by \cite{MoS} $W(r)$ can be expanded as an infinite power series in $r^{-1}$, which is however usually not convergent~\cite{GrGrHaSi}. The asymptotic expansion of $W(r)$ together with the information that $W(r)$ is analytic does not immediately provide an information on $W'(r)$. For example the function $e^{-ar}\sin(e^{ar})$ is analytic, has 0 as an asymptotic series at infinity, but its derivative is $O(1)$. 
\end{remark}

\begin{remark}\rm 
An estimate on $W'(r)$ similar to that in \eqref{eq:derexp} has already been proved in~\cite[Section 3.3.3]{Le} for the case of two molecules with dipole moments. It was indeed mentioned for the first time in~\cite{Le} that only the non-degeneracy of one of the two ground states is sufficient.
\end{remark}

\begin{remark}\rm 
The method used for proving Theorem \ref{thm:new} was recently adapted in \cite{AL} in order to study the compactness of molecular paths between two local minima of the interaction energy, in the case of two rigid molecules.  
\end{remark}

\subsection*{Notation}
For two functions $f,g$ of $\di>0$, with $g$  real valued, we say that $f = O(g)$, if $f$ is a continuous function of $r$ and there exist constants $C,D$ such that if $\di\geq C$ then $\|f(\di)\| \leq D |g(\di)|$. If $f$ takes its values in $L^2$ or in the algebra of bounded operators $\mathcal{B}(L^2)$ on $L^2$,  then the continuity and the inequality are understood in terms of the respective norms. We also write $f = \Od^m(g)$ if $f=O(g)$ and $\frac{d^k f}{dr^k} = O(\frac{d^k g}{dr^k})$ for all $k \leq m$. In particular, if $g = \frac{1}{\di^n}$ with $n$ an integer, then $f=\Od^m(g)$ means that $\|\frac{d^k}{dr^k}f(r)\|\leq Cr^{-n-k}$ for all $k\leq m$.
Thus in our notation the estimates of Theorems \ref{thm:new} can be summarized as
$W(r)=- \sigma r^{-6} + \Od^2(r^{-7})$.

\subsection*{Acknowledgements} 
I.A. is grateful to Volker Bach for suggesting part of this problem and to Dirk Hundertmark for discussions at an early stage of the work and for a suggestion that led to the proof of the important equations \eqref{Ifnocutoff}, \eqref{eq:chizIder}. All authors are grateful to Semjon Wugalter  for useful remarks which have led us to include the spin into account. This project has received funding from the European Research Council (ERC) under the European Union's Horizon 2020 research and innovation programme (grant agreement MDFT No 725528 of M.L.).
The research of I.A. was supported by the German Science Foundation under Grant No.  CRC 1173.

\section{The Feshbach map}

\subsection{Definition of the Feshbach map}
Let $\Px$ be an orthogonal projection on $Y$ (the space defined in \eqref{def:YN1N2}), with $\Ran \Px \subset \cQ Y\cap D(H^a)$.
Let also $\Px^\bot=1-\Px$ and $\Hx^{a,\bot}:=\Px^\bot \Hx^a \Px^\bot$, where we recall that $\Hx^a$ was defined in \eqref{eq:symHam}.
Assume that
there exists $C>0$ such that
\begin{equation}\label{stability}
 \Hx^{a,\bot}-E(\di) \geq C.
\end{equation}
 Then $E(\di)$ is an eigenvalue of the Feshbach map $F_\Px(E(r))$, where 
\begin{equation}\label{FP}
F_\Px(\lambda)=(\Px \Hx \Px-\Px \Hx \Px^\bot (\Hx^{a,\bot}-\lambda)^{-1} \Px^\bot \Hx \Px)|_{\Ran  \Px}.
\end{equation}
Equation \eqref{FP} tells us that finding the ground state energy $E(\di)$ reduces to solving a nonlinear fixed point problem on the range of $\Px$. The proof is well known, see for example \cite{BFS} and~\cite[Sec.~2.3]{AL}, but for convenience of the reader we sketch it here.

Let $\psix$ be a ground state of $\Hx^a$ and write $E=E(\di)$. Then $(\Hx^a-E) \psix=0$. This gives that $ \Px^\bot (\Hx^a-E) \psix=0$. Therefore writing $\psix= \Px \psix+ \Px^\bot \psix$, we obtain that
$$  \Px^\bot (\Hx^a-E) \Px^\bot \psix=- \Px^\bot (\Hx^a-E) \Px \psix=-\Px^\bot \Hx \Px \psix,$$
where in the last step we used that $\Px^\bot \Px\psix=0$ and that $H^a P= H P$. Thus, using that, due to \eqref{stability}, $ (\Hx^{a,\bot}-E)$ is invertible  we obtain that
\begin{equation}\label{PPbot}
	\Px^\bot \psix=-\Px^\bot (\Hx^{a,\bot}-E)^{-1} \Px^\bot \Hx \Px \psix.
\end{equation}
Multiplying both sides with $\Px \Hx$ we find
$$ \Px \Hx  \Px^\bot \psix=-\Px\Hx \Px^\bot (\Hx^{a,\bot}-E)^{-1} \Px^\bot \Hx \Px \psix.$$
Using $\Px^\bot=1-\Px$ on the left hand side of the last equation and $\Px\Hx\psix = E\Px\psix$ it follows that
$$ E \Px \psix = \Px\Hx \Px \psix-\Px\Hx \Px^\bot (\Hx^{a,\bot}-E)^{-1} \Px^\bot \Hx \Px \psix.$$
We next observe that $\Px \psix \neq 0$, otherwise the right hand side of \eqref{PPbot} would be zero and this would give that $\Px^\bot \psix$ is also zero, contradicting that $\psix \neq 0$. It therefore follows that $E$ is an eigenvalue of the Feshbach map.

\subsection{Choice of the projection $\Px$ in the spinless case}
 Let $\chi_1: \mathbb{R}^3
 \rightarrow \mathbb{R}$ be a spherically symmetric $C^\infty$ function
 supported in the ball $B(0,\frac{1}{6})$ and equal to $1$ in the ball
 $B(0,\frac{1}{7})$ with $ 0 \leq \chi_1 \leq 1$.  
Our goal is to compute the derivative $W'(\di)$. To do this we analyze $W(\di)$ for $\di$ near some $\di_0$. Note that it is convenient to cut the ground state off. The cut-off must be $\di$ dependent when we want to study the asymptotic behaviour as $\di\rightarrow\infty$. This would however introduce extra terms when we differentiate with respect to $\di$, something which we would like to avoid. Therefore, to study the interaction energy near $\di$ we choose a $\di_0 \in \R$ with $\di_0 < \di < \di_0 +1$ and we define
\begin{equation}\label{def:phi}
\phi_j\left(x_1,\dots,x_{N_j}\right) := c_j \left(\prod_{i=1}^{N_j}\chi_1\left(\frac{x_i}{\di_0}\right)\right) \groundst_j\left(x_1,\dots,x_{N_j}\right),
\end{equation}
so that the cut-off does not change when we vary $\di$ a little bit and where $c_j>0$ is chosen so that $\|\phi_j\|_{L^2}=1$.

We further introduce the resolvent
\begin{equation}\label{def:Rr}
\Rrx := \left(\Px_0^{\bot} \Hxabot_0 \Px_0^{\bot} - E(\infty)\right)^{-1} ,
\end{equation}
where
\begin{equation*}
\Hxabot_0 = \Pxbot_0\cQ_1\cQ_2\Hx_0\Pxbot_0,
\end{equation*}
with $\Hx_0$  defined in \eqref{Hdec} and
\begin{equation}\label{def:Px0}
\Px_0 := | \ppy\rangle \langle\ppy|.
\end{equation}
Here $\phi_{j,\di}(x_1,\dots,x_{N_j}) := \phi_j(x_1-\di e_1,\dots,x_{N_j}-\di e_1)$ is the cut-off ground state $\groundst_j$ with the nucleus placed at $\di e_1$. This simplifies our analysis considerably, since the cut-off ensures that $\ppy$ and $T_{\pi}\ppy$ have disjoint supports if $\pi\notin\subij$. Here $\subij$ is the subgroup of $S_{N_1+N_2}$ leaving the sets $\{1,\dots,N_1\}$ and $\{N_1+1,\dots,N_1+N_2\}$ invariant.

We further define
\begin{equation*}\label{chiR}
\chix\left(x_1,\dots,x_{N_1+N_2}\right):=\prod_{i=1}^{N_1}\left(\chi_1\left(\frac{6x_i}{7\di_0}\right)\right)\prod_{j=N_1+1}^{N_1+N_2} \left(\chi_1\left(\frac{6(x_j-\di e_1)}{7\di_0}\right)\right).
\end{equation*}
The dilation that we applied by multiplying with $\frac{6}{7}$ ensures that for $\di$ large enough
\begin{equation}
\label{eq:cutoffsupp}
\left\{
\begin{aligned}
	&\chix = 1 \text{ on } \supp(\ppy), \\
	&\chix T_\pi(\ppy) = 0, \qquad\forall \pi\in S_{N_1+N_2}\backslash \subij.
\end{aligned}\right.
\end{equation}

In the rest of the paper we will always assume without explicitly mentioning it that $\di$ is large enough. We choose 
\begin{equation}\label{eq:projP}
 \Px= \frac{1}{\|\cQ\psix\|^2}|\cQ\psix \rangle \langle \cQ\psix|,
 \end{equation} 
 where 
 \begin{equation}\label{def:psi0}
\psix := \frac{\psix_0}{\|\psix_0\|}, \text{ with }	 \psix_0 := \ppy -\chix\Rrx \Ix \ppy.
 \end{equation}
 Note that 
 \begin{equation}\label{eq:psicapT12psi}
 \supp(\psix) \cap \supp (T_{\pi}\psix) = \emptyset, \qquad\forall \pi\in S_{N_1+N_2}\backslash \subij,
 \end{equation}
 so $\psix$ and $T_{\pi}\psix$ are orthogonal, when $\pi\notin(\subij)$. Moreover,
 \begin{equation}\label{eq:psicapT12psi'}
 \psix=(-1)^\pi T_{\pi}\psix, \qquad \forall \pi\in \subij.
 \end{equation}
In \cite{A}, see also \cite{AS}, it was shown that \eqref{eq:neutralsplit} implies that there exist $C,c>0$, such that if $\di>c$ then
\begin{multline*}
\!\!\!\!\!\left(1-\frac{|\cQ(\ppy)\rangle\langle\cQ(\ppy)|}{\|\cQ(\ppy)\|^2}\right)\Hx^a\left(1-\frac{|\cQ(\ppy)\rangle\langle\cQ(\ppy)|}{\|\cQ(\ppy)\|^2}\right) \\ \geq C+E(\di).
\end{multline*}
Moreover, $\Hx^{a,\bot}:=\Pxbot\Hx^a\Pxbot$ is a perturbation of $$\left(1-\frac{|\cQ(\ppy)\rangle\langle\cQ(\ppy)|}{\|\cQ(\ppy)\|^2}\right)\Hx^a\left(1-\frac{|\cQ(\ppy)\rangle\langle\cQ(\ppy)|}{\|\cQ(\ppy)\|^2}\right),$$ of the order $O(\di^{-3})$ since
$\|\chix\Rrx\Ix\ppy\|_{H^2} = O(r^{-3})$.
Therefore, the following lemma holds.
 \begin{lemma}\label{lem:IMScor}
 	With the above choice of $\Px$ in \eqref{eq:projP}, the assumption \eqref{eq:neutralsplit} implies that there exist $C,c>0$ such that if $\di \geq c$ then 
 	\begin{equation}\label{eq:IMScor}
 	\Hx^{a,\bot}-E(\di) \geq C.
 	\end{equation}
 \end{lemma}

\section{Proof  of Theorem~\ref{thm:new}}\label{sec:proof}
In this section we prove Theorem \ref{thm:new}. 
We first observe that due to \eqref{eq:projP},  \eqref{eq:psicapT12psi} and \eqref{eq:psicapT12psi'}
\begin{equation}\label{eq:PHP}
\Px\Hx\Px = \frac{1}{\|\cQ\psix\|^2}\langle\cQ\psix,\Hx\cQ\psix\rangle\Px= \langle\psix,\Hx\psix\rangle\Px,
\end{equation}
where we have also used that $\Hx$ is a local operator.

Since $E(\di)$ is the ground state energy of $\Hx^a$, it is also an eigenvalue of the Feshbach map $F_P(E(\di))$. Therefore, recalling that $P$ has  rank 1  we find
$E(\di)=\frac{1}{\|\cQ\psix\|^2}\langle\cQ\psix, F_P(E(\di)) \cQ\psix \rangle $. As a consequence, using \eqref{eq:PHP} we find
\begin{equation}\label{eq:Erexp}
  \Imp(\di,E(\di))=0,
\end{equation}
where
\begin{equation}\label{eq:ImprE}
\Imp(\di,E) =E- \langle\psix,\Hx\psix\rangle + A,
\end{equation}
and 
\begin{equation}\label{def:A}
A=A(r,E):= \frac{1}{\|\cQ\psix\|^2}\left\langle \Px^\bot\Hx\cQ\psix, \left(\Hx^{a,\bot}-E\right)^{-1}\Px^\bot\Hx\cQ\psix \right\rangle.
\end{equation}
We will prove that $\Imp(\di,E)$ is twice continuously differentiable near  $E(\di)$ and we will apply the implicit function theorem to obtain information about the first two 	 derivatives of $E(\di)$. 

\subsection{Estimate of $\langle\psix,\Hx\psix\rangle$}\label{sec:lin}
We will now prove that
\begin{equation}\label{eq:psiHpsi}
\langle\psix,\Hx\psix\rangle = E(\infty) -\frac{\sigma}{\di^6} + \Odi\left(\frac{1}{\di^7}\right).
\end{equation}
Let $z_i=x_i, z_j=x_j-\di e_1$ for $i =\{1,\dots,N_1\}$ and 
$j=\{N_1+1\dots,N_1+N_2\}$, respectively. The variables $z_i$ are the relative coordinates of the electrons with respect to the nuclei. Below we add a tilde in the notation to indicate that we work in the variables $z_i$. We define
\begin{equation}
\psiz(z_1,\dots,z_{N_1+N_2}) = \psix(z_1,\dots,z_{N_1},z_{N_1+1} +\di e_1,\dots,z_{N_1+N_2}+\di e_1),
\end{equation}
and 
\begin{equation}\label{eq:decom}
\Hz = \Hz_0 + \Iz,
\end{equation}
with
\begin{multline*}\label{def:Hz0}
\Hz_0 = \sum_{i=1}^{N_1}\left(-\Delta_{z_i} - \frac{N_1}{|z_i|}\right) + \sum_{1\leq k < l \leq N_1}\left(\frac{1}{|z_k - z_l|}\right)  \\+ \sum_{j=N_1+1}^{N_1+N_2}\left(-\Delta_{z_j} - \frac{N_2}{|z_j|}\right) + \sum_{N_1+ 1\leq m < n \leq N_1 +N_2}\left(\frac{1}{|z_m - z_n|}\right),
\end{multline*}
\begin{equation}\label{def:Iz}
\Iz = \sum_{i=1}^{N_1}\sum_{j=N_1+1}^{N_1+N_2}\Izij, \quad
\Izij = \frac{1}{\di} + \frac{1}{|-z_i+\di e_1 + z_j|} - \frac{1}{|\di e_1 + z_j|} - \frac{1}{|\di e_1 -z_i|}.
\end{equation}
It then follows that
\begin{equation}\label{eq:PHE0P}
 \langle\psix,(\Hx-E(\infty))\psix\rangle = \langle\psiz,(\Hz-E(\infty))\psiz\rangle.
\end{equation}
Note that $\psiz = \|\psiz_0\|^{-1}\psiz_0$, where
\begin{equation}\label{def:psiz0}
\psiz_0 = \pp - \chiz\Rz\Iz\pp,
\end{equation}
with
\begin{equation*}
\chiz(z_1,\dots,z_{N_1+N_2}) = \chix(z_1,\dots,z_{N_1},z_{N_1+1} +\di e_1,\dots,z_{N_1+N_2}+\di e_1)
\end{equation*}
and
\begin{equation}
\Rz = (\Pz_0^{\bot}\Hz_0\Pz_0^{\bot}-E(\infty))^{-1}, \qquad 
\Pz_0 = |\pp\rangle\langle\pp|.\label{def:Pz0}
\end{equation}
Therefore, from \eqref{eq:PHE0P} and \eqref{def:psiz0} we find
\begin{equation}\label{eq:N1N2N3}
\langle\psix,(\Hx-E(\infty))\psix\rangle =\frac{1}{\|\psiz_0\|^2}(L_1+L_2+L_3),
\end{equation}
where
\begin{align*}
L_1 &= \langle\pp,(\Hz-E(\infty))\pp\rangle,
\\L_2 &= -2Re\langle\chiz\Rz \Iz\pp,(\Hz-E(\infty))\pp\rangle,
\\L_3 &= \langle\chiz\Rz \Iz\pp, (\Hz-E(\infty))\chiz\Rz \Iz\pp\rangle.
\end{align*}
Note that \eqref{eq:cutoffsupp} implies that 
\begin{equation}\label{eq:suppchiz}
\chiz=1 \text{ on } {\rm supp}(\pp),
\end{equation}
and therefore
\begin{equation}\label{eq:N2}
L_2 = -2Re\langle\Rz \Iz\pp,(\Hz-E(\infty))\pp\rangle.
\end{equation}
In $L_1,L_2,L_3$ the only $\di$-dependence is in the term $\Iz$, as the cut-off of the ground states is $\di$-independent, see \eqref{def:phi}. We next prove a lemma whose statements we will repeatedly need throughout the proof.
Before we state it we note that Newton's theorem~\cite[Sec.~9.7]{LL} implies~\cite[Lem.~5.1]{A} that
\begin{equation}\label{eq:Ipp}
\left\langle \pp , \Iz \pp   \right\rangle = 0.
\end{equation}
\begin{lemma}\label{lem:Ipsi0est}
	Let $f$ be as in \eqref{def:fij}. Then
	\begin{equation}\label{Ifnocutoff}
	\left(\Iz - \frac{1}{\di^3}f\right)\pp =\Odi\left(\frac{1}{\di^4}\right)
	\end{equation}
	and
	\begin{equation}
	\left\|\frac{d^\alpha}{d\di^\alpha}\chiz\Iz\right\|_{L_\infty} = O\left(\frac{1}{\di^{\alpha+1}}\right), \label{eq:chizIder}
	\end{equation}
for all $\alpha\geq0$. Moreover	
	\begin{equation}\label{eq:almosteig}
	\Hz_0 \pp = E(\infty) \pp + \Odi(e^{-c\di}) \text{ for some } c >0
	\end{equation}
	and
	\begin{equation}\label{est:HzE0}
	\left(\Hz-E(\infty)\right)\psiz = \Od^\infty\left(\frac{1}{\di^4}\right).
	\end{equation}
\end{lemma}
\begin{proof}
	We first prove \eqref{Ifnocutoff}. Note that we do not have to deal with the singularities of $\Iz$ as  they are away from the support of $\pp$.
	Using Taylor's theorem we obtain that for all $z  \in \mathbb{R}^3, r>0$ with  $|z| \leq \frac{|\di|}{2}$ we have
	\begin{equation}\label{eq:proimplicit}
		\frac{1}{|\di e_1-z|}= \frac{1}{\di}+\frac{z \cdot e_1}{\di^2}+\frac{3(z \cdot e_1)^2-|z|^2}{2 \di^3} + g(z,\di)
	\end{equation}
	where
	\begin{equation}\label{Appeq:gzr}
		g(z,\di)= \int_0^1 \!ds_1 \int_0^{s_1} \!ds_2 \int_0^{s_2} \!ds_3 \frac{15(z \cdot \widehat{(\di e_1-s_3 z)})^3-
			9(z \cdot \widehat{(\di e_1-s_3 z)}) |z|^2)}{|\di e_1-s_3 z|^4}.
	\end{equation}
	Using  \eqref{eq:proimplicit} together with \eqref{def:Iz} and \eqref{def:fij} one can show that
	\begin{multline*}
		(\Iz - \frac{1}{\di^3}f) \pp (z_1,\dots,z_{N_1+N_2})\\= \sum_{i=1}^{N_1}\sum_{j=N_1+1}^{N_1+N_2} (g(z_i,\di)+g(-z_j,\di)-g(z_i-z_j,\di) )\pp.
	\end{multline*} 
	Using the last equality and \eqref{def:phi} we arrive at 
	\begin{multline*}
		\frac{d^\alpha}{d \di^\alpha}  \left[(\Iz - \frac{1}{\di^3}f)\pp\right](z_1,\dots,z_{N_1+N_2}) \\ 
		\label{eq:intrest} = \sum_{i=1}^{N_1}\sum_{j=N_1+1}^{N_1+N_2} \left[\frac{d^\alpha}{d \di^\alpha} \Big(g(z_i,\di)+g(-z_j,\di)-g(z_i-z_j,\di) \Big)\right]\pp,
	\end{multline*}
	for every integer $\alpha \geq0$.
	But applying the dominated convergence theorem we obtain that  
	\begin{equation*}\label{Appeq:gzrder}
		\left[\frac{d^\alpha}{d \di^\alpha} \big(g(z_i,\di)+g(-z_j,\di)-g(z_i-z_j,\di) \big)\right]\pp = O\left(\frac{1}{\di^{4+\alpha}}\right)
	\end{equation*}
	where the exponential decay of $\groundst$, namely \eqref{eq:expdec}, is needed for establishing \eqref{Appeq:gzrder}. It follows that
	\begin{equation*}
		\frac{d^\alpha}{d\di^\alpha}\left[(\Iz-\frac{1}{\di^3}f)\pp\right] = O\left(\frac{1}{\di^{4+\alpha}}\right),
	\end{equation*}
 so that we  arrive at \eqref{Ifnocutoff}.

 The estimate \eqref{eq:chizIder} can be proven in a similar manner. The only difference is that instead of \eqref{eq:proimplicit} and \eqref{Appeq:gzr} we use
	\begin{equation*}
		\frac{1}{|\di e_1 - z|} = \frac{1}{\di} +\int_{0}^{1}\frac{z\cdot\widehat{(\di e_1-s_1 z)}}{|\di e_1-s_1 z|^2}ds_1.
	\end{equation*}

In order to prove \eqref{eq:almosteig} we observe that
	\begin{equation}\label{eq:proofalmosteig}
		(\Hz_0-E(\infty))\pp = \frac{-2\nabla\chiz\nabla(\gsi\otimes\gsii) -\Delta\chiz(\gsi\otimes\gsii)}{\|\chiz(\gsi\otimes\gsii)\|}.
	\end{equation}
	Since $\nabla\chiz$ is supported $O(\di)$ far from the nuclei of the atoms, using \eqref{eq:expdec} and \eqref{eq:proofalmosteig} we find that
	\begin{equation*}
	\|(\Hz_0-E(\infty))\pp\|_{H^1} =O(e^{-cr}).
	\end{equation*}
	Since moreover, if we vary $\di$ on the left hand side the latter does not change, we arrive at \eqref{eq:almosteig}.

	We now prove \eqref{est:HzE0}. From \eqref{eq:decom} we obtain that
	\begin{multline}\label{eq:HzE0psiz}
	(\Hz-E(\infty))\psiz_0 = (\Hz_0-E(\infty))\pp - (\Hz_0-E(\infty))\chiz\Rz\Iz\pp\\ + \Iz\pp - \Iz\chiz\Rz\Iz\pp.
	\end{multline}
	Therefore using \eqref{Ifnocutoff}, \eqref{eq:chizIder} and \eqref{eq:almosteig} we find
	\begin{equation}\label{eq:HminusEinf}
	(\Hz-E(\infty))\psiz_0 =  - (\Hz_0\Pz_0^\bot-E(\infty))\chiz\Rz\Iz\pp + \Iz\pp +\Od^\infty\left(\frac{1}{\di^4}\right),
	\end{equation}
	where we recall that $\Pz_0$ was defined in \eqref{def:Pz0}, and we have also used that $\Iz\pp = \Pz_0^\bot\Iz\pp$, see \eqref{eq:Ipp}, and that $[\Pz_0,\chiz\Rz] = 0$. Since $\chiz\Rz = \Rz\chiz + [\chiz,\Rz]$ we find using the equality
	
	\begin{equation}\label{eq:Hz0Pzbot}
	(\Hz_0\Pz_0^\bot-E(\infty))\Rz = 1 + \Pz_0\Hz_0\Pzbot_0\Rz = 1 + \Od^\infty(e^{-c\di}),
	\end{equation}
	that
	\begin{multline}
	\label{eq:commchizRz}
	(\Hz_0\Pz_0^\bot-E(\infty))\chiz\Rz\Iz\pp =  \chiz\Iz\pp\\ + (\Hz_0\Pz_0^\bot-E(\infty))[\chiz,\Rz]\Iz\pp + \Od^\infty(e^{-c\di}).
	\end{multline}  
	From \eqref{eq:suppchiz} we obtain that 
	\begin{equation}
	\label{eq:chizIz}
	\chiz\Iz\pp = \Iz\pp.
	\end{equation}
	Using \eqref{eq:HminusEinf}, \eqref{eq:commchizRz} and \eqref{eq:chizIz} we infer
	\begin{equation*}
	(\Hz-E(\infty))\psiz_0 = - (\Hz_0\Pz_0^\bot-E(\infty))\Rz\Pz_0^\bot[\chiz,\Hz_0]\Pz_0^\bot\Rz\Iz\pp +\Od^\infty\left(\frac{1}{\di^4}\right).
	\end{equation*}
	Since $[\chiz,\Hz_0]\Rz = O\left(\di^{-1}\right)$ is constant in a neighborhood of $r_0$, we arrive with the help of \eqref{Ifnocutoff} at
	\begin{equation}
	\label{eq:HzEpsiz0}
	(\Hz-E(\infty))\psiz_0 = \Odi\left(\frac{1}{\di^4}\right).
	\end{equation}	
	Using \eqref{def:psiz0} and \eqref{Ifnocutoff} we obtain that
	\begin{equation}
	\label{est:psi0norm}
	\|\psiz_0\|^2 = 1 + \Od^\infty\left(\frac{1}{\di^6}\right),
	\end{equation}
	where we have used the orthogonality of the two summands on the right hand side of \eqref{def:psiz0}. Using \eqref{est:psi0norm} and
	  \eqref{eq:HzEpsiz0}  we conclude the proof of \eqref{est:HzE0} hence of Lemma \ref{lem:Ipsi0est}.
	\end{proof}

We next estimate $L_1,L_2,L_3$.
Using \eqref{eq:decom}, \eqref{eq:Ipp} and \eqref{eq:almosteig} we find that
\begin{equation}
\label{eq:PHPN1}
L_1 = \Odi(e^{-c\di}), \text{ for some } c>0.
\end{equation}
Now we want to show that 
\begin{equation}\label{eq:PHPN2}
L_2 = -2\frac{1}{\di^6}\langle f\pp, \Rz f\pp\rangle + \Odi\left(\frac{1}{\di^7}\right).
\end{equation}
Using \eqref{eq:decom} and \eqref{eq:N2} we find
\begin{equation*}
L_2 = -2 \langle\Iz\pp, \Rz\Iz\pp\rangle - 2\Re \langle\Rz\Iz\pp, (\Hz_0 - E(\infty)) \pp\rangle.
\end{equation*}
Further using \eqref{Ifnocutoff} and \eqref{eq:almosteig} we can show that
\begin{equation*}
L_2 = -2 \langle\Iz\pp,\Rz\Iz\pp\rangle+O^\infty_d(e^{-cr}).
\end{equation*}
To arrive at \eqref{eq:PHPN2}, it remains for us to prove that
\begin{equation}\label{eq:IRI}
\langle\Iz\pp,\Rz\Iz\pp\rangle - \frac{1}{\di^6}\langle f\pp, \Rz f\pp\rangle = \Odi\left(\frac{1}{\di^7}\right).
\end{equation}
Indeed, observe that the left hand side of \eqref{eq:IRI} can be split into
\begin{equation*}
\langle\Iz\pp-\frac{f}{\di^3}\pp,\Rz\Iz\pp\rangle + \langle\frac{f}{\di^3}\pp,\Rz(\Iz\pp-\frac{f}{\di^3}\pp)\rangle.
\end{equation*}
Using \eqref{Ifnocutoff}, the estimate \eqref{eq:IRI} follows and therefore we obtain \eqref{eq:PHPN2}.

We now estimate $L_3$. With the help of Leibniz rule for the kinetic part of $\Hz$ and \eqref{eq:decom}, we obtain
\begin{equation}\label{def:PHPN3}
L_3 = L_{31} + L_{32} + L_{33} ,
\end{equation}
where
\begin{align*}
L_{31} &=   \langle\chiz\Rz\Iz\pp, \chiz(\Hz_0-E(\infty))\Rz\Iz\pp\rangle,\label{def:PHPN31}
\\L_{32} &= \langle\chiz\Rz\Iz\pp, \Iz\chiz\Rz\Iz\pp\rangle - \langle\chiz\Rz\Iz\pp, (\Delta\chiz)\Rz\Iz\pp\rangle,
\\L_{33} &= -2\langle\chiz\Rz\Iz\pp, \nabla\chiz\cdot\nabla\Rz\Iz\pp\rangle.
\end{align*}
Observe that \eqref{eq:Ipp} gives $\Iz\pp = \Pz_0^\bot\Iz\pp$ and moreover $\Pz_0^\bot$ commutes with $\Rz$ and $\chiz$, because of \eqref{def:Pz0} and \eqref{eq:suppchiz}, respectively. Thus, $\Hz_0$  can be replaced by $\Pz_0^\bot\Hz_0\Pz_0^\bot$. It follows that
\begin{equation*}
L_{31} = \langle\chiz\Rz\Iz\pp, \chiz\Iz\pp\rangle,
\end{equation*}
where we have also used that $(\Pz_0^\bot\Hz_0\Pz_0^\bot- E(\infty))\Rz=1$. Using \eqref{eq:suppchiz} once more we find that
\begin{equation*}
L_{31} = \langle\Iz\pp, \Rz\Iz\pp\rangle.
\end{equation*}
Hence with \eqref{eq:IRI} we come to the conclusion that
\begin{equation}
\label{eq:PHPN31}
L_{31} = \frac{1}{\di^6}\langle f\pp, \Rz f\pp\rangle + \Odi\left(\frac{1}{\di^7}\right).
\end{equation}
Observe now that
\begin{equation}\label{est:nablachiz}
|\nabla\chiz|\leq\frac{c}{\di}, \qquad |\Delta\chiz|\leq\frac{c}{\di^2},
\end{equation}
Using \eqref{Ifnocutoff}, \eqref{eq:chizIder} and \eqref{est:nablachiz} we find that
\begin{equation*}
\label{eq:PHPN32}
L_{32} = \Odi\left(\frac{1}{\di^7}\right), 
\qquad L_{33} = \Odi\left(\frac{1}{\di^7}\right),
\end{equation*}
where we have also used that  $\nabla\Rz$ is bounded.
Together with \eqref{def:PHPN3}, \eqref{eq:PHPN31} we find that 
\begin{equation}\label{eq:PHPN3}
L_3 = \frac{1}{\di^6}\langle f \pp, \Rz f \pp\rangle + \Odi\left(\frac{1}{\di^7}\right).
\end{equation}
From \eqref{eq:PHPN1}, \eqref{eq:PHPN2} and \eqref{eq:PHPN3} we find
\begin{equation*}\label{eq:PHPsumN}
L_1 + L_2 + L_3 = -\frac{1}{\di^6} \langle f\pp, \Rz f \pp\rangle + \Odi\left(\frac{1}{\di^7}\right).
\end{equation*}
It is known (see for example~\cite[Proof of Lemma 5.6]{A}) that
\begin{equation*}
\langle f \pp, \Rz f \pp \rangle =\sigma+O(e^{-cr}),
\end{equation*}
and since the left hand side is $\di$-independent, we find
\begin{equation*}\label{eq:PHPdoteqdsig}
\langle f \pp, \Rz f \pp \rangle = \sigma + \Od^\infty(e^{-c\di}).
\end{equation*}
We conclude that
\begin{equation}
\label{eq:PHPsig}
\left(L_1 + L_2 + L_3\right) = -\frac{\sigma}{\di^6} + \Odi\left(\frac{1}{\di^7}\right).
\end{equation}
From \eqref{eq:N1N2N3}, \eqref{est:psi0norm} and \eqref{eq:PHPsig} we arrive at \eqref{eq:psiHpsi}.

\subsection{Nonlinear term}
We now focus on the nonlinear term $A$ defined in \eqref{def:A}. To apply the implicit function theorem we first fix $E$ near $E(\di)$, such that
\begin{equation}
\label{eq:E-Er}
\left|E-E(\di)\right|<\frac{c}{\di^7}\qquad\text{for some } c>0,
\end{equation}
and we investigate the partial derivatives of $A$ with respect to $\di$. Below we will write $\frac{d}{d\di}$ for expressions  that do not depend on $E$ but only on $\di$. 

 Since $P^\bot \cQ \psi=0$ we find
\begin{equation}\label{eq:nonlin}
A=\frac{1}{\|\cQ\psix\|^2}\left\langle \Px^\bot(\Hx - E(\infty))\cQ\psix, \left(\Hxabot-E\right)^{-1}\Px^\bot(\Hx - E(\infty))\cQ\psix \right\rangle.
\end{equation}
Our goal is to prove that
\begin{equation}\label{est:Aalpha}
A = O\left(\frac{1}{\di^8}\right),
\quad \frac{\partial^\alpha}{\partial r^\alpha} A =
O\left(\frac{1}{\di^9}\right),\qquad  \text{ for } \alpha\in \{1,2\}.
\end{equation} 

Using that $\Px^\bot$ is a projection commuting with $\Hxabot$ and $\cQ$ is a projection commuting with $\Hx,\Pxbot$ and $\Hxabot$, we find 
\begin{equation*}
A =\frac{1}{\|\cQ\psix\|^2}\left\langle(\Hx - E(\infty))\psix, \left(\Hxabot-E\right)^{-1}\Px^\bot(\Hx - E(\infty))\cQ\psix \right\rangle.
\end{equation*}
From the fact that
\begin{equation*}
\Pxbot = 1 - \frac{|\cQ\psix\rangle\langle\cQ\psix|}{\|\cQ\psix\|^2},
\end{equation*}
and that
\begin{equation*}
\langle\cQ\psix,(\Hx-E(\infty))\cQ\psix\rangle = \langle\psix,(\Hx-E(\infty))\psix\rangle\|\cQ\psix\|^2,
\end{equation*}
we find
\begin{equation*}
A = - B + C,
\end{equation*}
with $B = M \langle\psix,(\Hx-E(\infty))\psix\rangle$, 
and
\begin{equation}\label{def:C}
C=\frac{1}{\|\cQ\psix\|^2}\langle(\Hx-E(\infty))\psix,(\Hxabot-E)^{-1}(\Hx-E(\infty))\cQ\psix\rangle.
\end{equation}
In the definition of $B$,
\begin{equation}\label{def:M}
M = \frac{1}{\|\cQ\psix\|^2} \langle (\Hx-E(\infty))\psix,(\Hxabot-E)^{-1}\cQ\psix\rangle.
\end{equation}
We now estimate $B$. Due to \eqref{eq:psiHpsi} we have 
\begin{equation}\label{def:B}
B = M\;\Odi\left(\frac{1}{\di^6}\right).
\end{equation}
We will show that 
\begin{equation}
\label{est:Mderalpha}
\frac{\partial^\alpha}{\partial\di^\alpha}M = O\left(\frac{1}{\di^4}\right), \qquad \text{for } \alpha \in\{0,1,2\}.
\end{equation}
To this end we will prove that
\begin{equation}
\label{est:Hxpsixder}
\frac{d^\alpha}{d\di^\alpha}(\Hx-E(\infty))\psix = O\left(\frac{1}{\di^4}\right), \qquad  \text{for } \alpha\in\{0,1\},
\end{equation}
and
\begin{equation}
\label{est:Hxpsixder2}
\frac{d}{d\di}(1-\Delta)^{-\frac{1}{2}}\frac{d}{d\di}(\Hx-E(\infty))\psix = O\left(\frac{1}{\di^4}\right).
\end{equation}
An elementary computation gives that
\begin{equation}
\label{eq:Hxtaur}
(\Hx-E(\infty))\psix = \tau_\di\left((\Hz-E(\infty))\psiz\right),
\end{equation}
where 
\begin{equation}\label{def:tau}
(\tau_\di\Phi)(x_1,\dots,x_{N_1+N_2}) := \Phi(x_1,\dots,x_{N_1},x_{N_1+1}-\di e_1,\dots,x_{N_1+N_2}-\di e_1).
\end{equation}
The following two lemmata are going to be useful

\begin{lemma}\label{lem:Mrder}
	Let $a,b \in \R$ with $a<b$. We consider $\Phi:(a,b) \rightarrow L^2$  differentiable with $\Phi(r) \in H^1(\R^{3(N_1+N_2)})$ for all $\di \in (a,b)$  and 
$$\vec{v}=(\underbrace{0,\dots,0}_{3N_1 \text{times}},\underbrace{-1,0,0,-1,0,0,\dots,-1,0,0}_{N_2 \text{times}})^\top.$$ 
Then $\tau_. \Phi(.):(a,b) \rightarrow L^2$ is differentiable and
	\begin{equation}\label{eq:taurPhider}
	\frac{d}{d\di}\left(\tau_\di \Phi(\di)\right) = (\vec{v}\cdot\nabla)\left(\tau_\di \Phi(\di)\right) + \tau_\di \left(\Phi'(\di)\right).
	\end{equation}
\end{lemma}

\begin{proof}
	To prove \eqref{eq:taurPhider} we first observe that
	\begin{multline*}
	\lim_{h \rightarrow 0} \frac{\tau_{\di+h}\Phi(\di+h) - \tau_\di \Phi(\di)}{h}\\ = \lim_{h \rightarrow 0} \left(\tau_{\di+h}\frac{\Phi(\di+h)-\Phi(\di)}{h} + \frac{\tau_{\di+h}-\tau_\di}{h}\Phi(\di)\right).
	\end{multline*}
	Since $\tau_{\di+h}$ is strongly continuous in $h$ we have that
	\begin{equation*}\label{eqlem:Phider}
	\lim_{h \rightarrow 0}\tau_{\di+h}\frac{\Phi(\di+h)-\Phi(\di)}{h} =\tau_\di\Phi'(\di).
	\end{equation*}
	Moreover, since the momentum operator is the generator of translations and $\Phi(\di)\in H^1$ we find 
	\begin{equation*}
	\label{eqlem:tauder}
	\lim_{h \rightarrow 0}\frac{\tau_{\di+h}-\tau_\di}{h}\Phi(\di) = (\vec{v}\cdot\nabla)\left(\tau_\di \Phi(\di)\right),
	\end{equation*}
	in the $L^2$-sense. Hence we have proved \eqref{eq:taurPhider}.
\end{proof}

\begin{lemma}\label{lem:PhiH1}
	\begin{equation*}\label{est:PhiH1}
	\|(\Hz -E(\infty))\psiz\|_{H^1} = O \left(\frac{1}{\di^4}\right).
	\end{equation*}
\end{lemma}

\begin{proof}
	We start with estimating some terms of the right hand side of \eqref{eq:HzE0psiz}.
	Using \eqref{eq:proofalmosteig} and \eqref{eq:expdec} we find that
	\begin{equation}
	\label{estlem:H1H0zpp}
	\left\|(\Hz_0-E(\infty))\pp\right\|_{H^1} = O(e^{-cr}),
	\end{equation} 
	because the support of the derivatives of $\chiz$ is far from the center of the atoms.
	Since $\chiz\Iz= O\left(\frac{1}{\di}\right)$, see \eqref{eq:chizIder}, and $\nabla(\chiz\Iz)= O\left(\frac{1}{\di^2}\right)$, using \eqref{Ifnocutoff} it follows that
	\begin{equation}
	\label{estlem:H1IchiRIpp}
	\left\|\Iz\chiz\Rz\Iz\pp\right\|_{H^1} = O\left(\frac{1}{\di^4}\right),
	\end{equation} 
	where due to the presence of the resolvent $\Rz$ it is enough that $\|\Iz\pp\|_{L^2}=O(r^{-3})$.	
	Using \eqref{eq:HzE0psiz}, \eqref{estlem:H1H0zpp} and \eqref{estlem:H1IchiRIpp} we obtain
	\begin{equation}
	(\Hz -E(\infty))\psiz = \Iz\pp - (\Hz_0-E(\infty))\chiz\Rz\Iz\pp+ O_{H^1}\left(\frac{1}{\di^4}\right),
	\end{equation}
	with the self-explanatory notation $O_{H^1}$.
	This gives with the help of Leibniz' rule
	\begin{multline*}\label{eqlem:Hzpsizfull}
	(\Hz -E(\infty))\psiz = \Iz\pp -\chiz(\Hz_0-E(\infty))\Rz\Iz\pp\\ + 2\nabla\chiz\cdot\nabla\Rz\Iz\pp + \Delta\chiz\Rz\Iz\pp + O_{H^1}\left(\frac{1}{\di^4}\right).
	\end{multline*}
	From \eqref{eq:suppchiz} and \eqref{eq:Hz0Pzbot} we find that
	\begin{equation*}\label{eqlem:chizHz0}
	\chiz(\Hz_0-E(\infty))\Rz\Iz\pp = \Iz\pp + \Pz_0\Hz_0\Pzbot_0\Rz\Iz\pp.
	\end{equation*}
	It follows that
	\begin{multline*}
	(\Hz -E(\infty))\psiz = -\Pz_0\Hz_0\Pzbot_0\Rz\Iz\pp + 2\nabla\chiz\cdot\nabla\Rz\Iz\pp\\ + \Delta\chiz\Rz\Iz\pp + O_{H^1}\left(\frac{1}{\di^4}\right).
	\end{multline*}
	Since $\Pz_0\Hz_0\Pzbot_0 = \Pz_0(\Hz_0-E(\infty))\Pzbot_0$, we find using \eqref{eq:almosteig} that
	\begin{equation}
	\label{eqlem:P0zH0z}
	\Pz_0\Hz_0\Pzbot_0 =O(e^{-cr}), \qquad \nabla\Pz_0\Hz_0\Pzbot_0 =O(e^{-cr}).
	\end{equation}
This concludes the proof of Lemma~\ref{lem:PhiH1}.
\end{proof}

We now continue with the proofs of \eqref{est:Hxpsixder} and \eqref{est:Hxpsixder2}.
Using \eqref{est:HzE0} and \eqref{eq:Hxtaur} together with Lemma \ref{lem:Mrder} for $\Phi = (\Hz -E(\infty))\psiz $, Lemma \ref{lem:PhiH1} and the translation invariance of the $L^2$ and $H^1$ norms we arrive at \eqref{est:Hxpsixder}.

To prove \eqref{est:Hxpsixder2} observe that from Lemma \ref{lem:Mrder} and \eqref{eq:Hxtaur} we have
\begin{multline*}\label{eq:lapHxEpsix}
(1-\Delta)^{-\frac{1}{2}}\frac{d}{d\di}(\Hx-E(\infty))\psix =(1-\Delta)^{-\frac{1}{2}}(\vec{v}\cdot\nabla)((\Hx-E(\infty))\psix\\ + \tau_\di\left((1-\Delta)^{-\frac{1}{2}}\frac{d}{d\di}\left((\Hz-E(\infty))\psiz\right)\right),
\end{multline*}
Thus, using  the boundedness and translation invariance of $(1-\Delta)^{-\frac{1}{2}}(\vec{v}\cdot\nabla)$ and \eqref{est:Hxpsixder}, we find that
\begin{equation*}
\label{est:dersignabHxE}
\frac{d}{d\di}(1-\Delta)^{-\frac{1}{2}}(\vec{v}\cdot\nabla)(\Hx-E(\infty))\psix = O\left(\frac{1}{\di^4}\right).
\end{equation*}
Using \eqref{est:HzE0} we obtain that
\begin{equation*}
\left\|(1-\Delta)^{-\frac{1}{2}}\frac{d}{d\di}\left((\Hz-E(\infty))\psiz\right)\right\|_{H^1} = O\left(\frac{1}{\di^5}\right)
\end{equation*}
and
\begin{equation*}
\left\| \frac{d^2}{d\di^2}\left((\Hz-E(\infty))\psiz\right) \right\|_{L^2} = O\left(\frac{1}{\di^6}\right).
\end{equation*}
Thus we can apply Lemma \ref{lem:Mrder} for $\Phi(\di) =(1-\Delta)^{-\frac{1}{2}}\frac{d}{d\di}(\Hz-E(\infty))\psiz$ and find that
\begin{equation*}
\label{est:tauHzpsizder}
\frac{d}{d\di}\tau_\di\left((1-\Delta)^{-\frac{1}{2}}\frac{d}{d\di}\left((\Hz-E(\infty))\psiz\right)\right) = O\left(\frac{1}{\di^5}\right).
\end{equation*}
This proves \eqref{est:Hxpsixder2}.

Using that $\psix =\tau_\di\psiz$ and the $\di$-independence of $\psiz$, we can apply Lemma \ref{lem:Mrder} for $\Phi = \psiz$ and find that 
\begin{equation*}\label{est:psixder}
\frac{d^\alpha}{d\di^\alpha}\psix =\tau_\di\left((\vec{v}\cdot\nabla)^\alpha\psiz\right) = O\left(1\right) \qquad\text{for }\alpha\in\{0,1,2\}.
\end{equation*}
In fact this can be done without Lemma \ref{lem:Mrder} using that the momentum operator is the generator of translations.
In a similar manner we find that
\begin{equation*}
\frac{d^\alpha}{d\di^\alpha}T_\pi\psix = O\left(1\right), \text{ } \alpha\in\{0,1,2\} \qquad 
\frac{d^\alpha}{d\di^\alpha}(\Hx-E(\infty))T_\pi\psix = O\left(\frac{1}{r^4}\right),  \text{ } \alpha\in\{0,1\}.
\end{equation*}
\begin{equation*}
\frac{d}{d\di}(1-\Delta)^{-\frac{1}{2}}\frac{d}{d\di}(\Hx-E(\infty)) T_\pi \psix = O\left(\frac{1}{\di^4}\right).
\end{equation*}
Thus, from the definition of $\cQ$, 
we obtain 
\begin{equation}\label{est:Qpsix}
\frac{d^\alpha}{d\di^\alpha}\cQ\psix = O(1), \quad \alpha\in\{0,1,2\} 
\end{equation}
and 
\begin{equation}
\label{est:QHxpsixderalpha}
\frac{d^\alpha}{d\di^\alpha}\cQ(\Hx-E(\infty))\psix = O\left(\frac{1}{\di^4}\right), \quad\text{for }\alpha\in\{0,1\},
\end{equation}
as well as
\begin{equation}
\label{est:QHxpsixder2}
\frac{d}{d\di}(1-\Delta)^{-\frac{1}{2}}\frac{d}{d\di}\cQ(\Hx-E(\infty))\psix = O\left(\frac{1}{\di^4}\right).
\end{equation}
We now continue with estimating $M$, defined in \eqref{def:M}. We note that
\begin{equation}
\label{est:M}
M= O\left(\frac{1}{\di^4}\right),
\end{equation}
 because  
\begin{equation}\label{est:HxEpsix}
\|(\Hx-E(\infty))\psix\| = \|(\Hz-E(\infty))\psiz\| \overset{\eqref{est:HzE0}}{=} O\left(\frac{1}{\di^4}\right),
\end{equation}
and
\begin{equation}
\label{eq:normQpsi}
\frac{1}{\|\cQ\psix\|^2} = \binom{N_1+N_2}{N_1}.
\end{equation}
Using \eqref{est:Hxpsixder} and \eqref{est:Qpsix} we find that
\begin{equation}\label{eq:Mder}
\frac{\partial}{\partial\di}M =  O\left(\frac{1}{\di^4}\right) + \frac{1}{\|Q \psi\|^2} \langle(\Hx-E(\infty))\psix,\left[\frac{\partial}{\partial\di}(\Hxabot-E)^{-1}\right]\cQ\psix\rangle.
\end{equation}
Writing the difference quotients for the partial derivative of the resolvent and using the second resolvent formula it follows that
\begin{equation}\label{eq:Resolder}
\left[\frac{\partial}{\partial\di}\left(\Hxabot-E\right)^{-1}\right]\cQ = \left(\Hxabot-E\right)^{-1}\left(\frac{d }{d\di}(\Pxbot\Hx\Pxbot)\right)\left(\Hxabot-E\right)^{-1}\cQ,
\end{equation}
where we have also used that the orthogonal projection $\cQ$ commutes with $\Hxabot$.
We now observe that
\begin{multline}\label{eq:normHxbotder}
\left\|(1-\Delta)^{-\frac{1}{2}}\left(\frac{d}{d\di}(\Pxbot\Hx\Pxbot)\right)(1-\Delta)^{-\frac{1}{2}}\right\|_{B(L^2)}
\\ = \left\|(1-\Delta)^{-\frac{1}{2}} \left( -\frac{d\Px}{d\di}\Hx\Pxbot +\Pxbot\frac{d\Hx}{d\di}\Pxbot -\Pxbot\Hx\frac{d\Px}{d\di} \right) (1-\Delta)^{-\frac{1}{2}}\right\|_{B(L^2)}.
\end{multline}
Moreover using \eqref{Hdec} and the definition of $\Ix$, we find
\begin{equation}\label{eq:Hxder}
\frac{d\Hx}{d\di} =  -\frac{N_1 N_2}{\di^2} + B_{\di},
\end{equation}
where
\begin{equation*}
B_{\di} = \sum_{j=1}^{N_1+N_2}\frac{N_2(x_j-\di e_1)\cdot e_1}{|x_j-\di e_1|^3},
\end{equation*}
and therefore due to Hardy's inequality
\begin{equation}\label{est:Hxder}
\left\|(1-\Delta)^{-\frac{1}{2}}\frac{d \Hx}{d\di}(1-\Delta)^{-\frac{1}{2}}\right\|_{B(L^2)} = O(1).
\end{equation}
Note that arguing as in the proof of \eqref{est:Qpsix} we obtain that
\begin{equation}
\label{est:Pxder}
\left\|\frac{d\Px}{d\di}(1-\Delta)^{\frac{1}{2}}\right\|_{B(L^2)} = O(1).
\end{equation}
Due to the $-\Delta$ form boundedness of $\Hx$, uniformly in $\di$, it follows from \eqref{eq:normHxbotder}, \eqref{est:Hxder} and \eqref{est:Pxder} that
\begin{equation*}\label{est:Hxbotder}
\left\|(1-\Delta)^{-\frac{1}{2}}\left( \frac{d}{d\di}(\Pxbot\Hx\Pxbot) \right)(1-\Delta)^{-\frac{1}{2}}\right\|_{B(L^2)} = O(1).
\end{equation*}
Since moreover $-\Delta$ is $\Hx^\bot$ form bounded, uniformly for large $\di$, we find that
\begin{equation*}\label{est:Hxabot}
\left\|(1-\Delta)^{\frac{1}{2}}\cQ\left(\Hxabot-E\right)^{-\frac{1}{2}}\right\|_{L^2} = O(1).
\end{equation*}
From \eqref{eq:Resolder},  we find that
\begin{equation}\label{est:Resolder}
\left[\frac{\partial}{\partial \di}\left(\Hxabot-E\right)^{-1}\right]\cQ = O(1).
\end{equation}
Using \eqref{est:HxEpsix}, \eqref{eq:normQpsi}, \eqref{eq:Mder} and \eqref{est:Resolder} we find
\begin{equation}\label{est:Mder}
\frac{\partial M}{\partial \di} = O\left(\frac{1}{\di^4}\right).
\end{equation}
If we try to differentiate $\frac{\partial M}{\partial \di}$ with respect to $\di$ we run into the problem that $\frac{d^2\Hx}{d\di^2}$ is not in $L^1_{loc}$. To remedy this we write the difference quotient for $\frac{d\Hx}{d\di}$ and perform changes of variables so that we do not have to differentiate $\frac{d\Hx}{d\di}$. To this end using \eqref{est:Hxpsixder}, \eqref{est:Hxpsixder2} and \eqref{est:Qpsix} we can argue similarly as in the proof of \eqref{est:Mder} to find
\begin{multline*}
\label{est:partialMder}
\frac{\partial^2 M}{\partial\di^2} = O\left(\frac{1}{\di^4}\right)\\ + \frac{\partial}{\partial\di} \left\langle(\Hx-E(\infty))\psix, (\Hxabot-E)^{-1}\Pxbot\frac{d\Hx}{d\di}\Pxbot(\Hxabot-E)^{-1}\cQ\psix \right\rangle.
\end{multline*}
This can be rewritten as
\begin{equation}
\label{est:MPhiPsi}
\frac{\partial^2 M}{\partial\di^2} =  O\left(\frac{1}{\di^4}\right) + \frac{\partial}{\partial\di}\left\langle\Phi,\frac{d\Hx}{d\di}\Psi\right\rangle,
\end{equation}
where
\begin{equation*}
\Phi = \Pxbot\left(\Hxabot-E\right)^{-1}\cQ\left(\Hx -E(\infty)\right)\psix\quad\text{and}\quad \Psi = \Pxbot\left(\Hxabot-E\right)^{-1}\cQ\psix
\end{equation*}
both belong to $H^2$. Note that we could add $\cQ$ in the definition of $\Phi$ because $\cQ$ is an orthogonal projection commuting with $\Pxbot,\Hxabot,\Hx,\frac{d\Hx}{d\di}$.
Using \eqref{est:Qpsix}, \eqref{est:QHxpsixderalpha}, \eqref{est:Pxder} and \eqref{est:Resolder} we find that
$$\frac{\partial^\alpha}{\partial\di^\alpha}\Phi = O\left(\frac{1}{\di^4}\right),\qquad 
\frac{\partial^\alpha}{\partial\di^\alpha}\Psi = O\left(\frac{1}{\di^4}\right)$$
for $\alpha\in\{0,1\}$.
We arrive at
\begin{equation}\label{eq:Mder2}
\frac{\partial^2 M}{\partial\di^2} = O\left(\frac{1}{\di^4}\right) + \lim_{h \rightarrow 0}  \left\langle\Phi,\frac{1}{h}\left(\left.\frac{d \Hx}{d\di}\right|_{\di+h}-\left.\frac{d\Hx}{d\di}\right|_{\di}\right)\Psi\right\rangle.
\end{equation}
Observe that by \eqref{eq:Hxder} it is enough to prove that
\begin{equation}\label{est:limBr}
\lim_{h \rightarrow 0} \left\langle\Phi, \frac{B_{\di+h}-B_{\di}}{h} \Psi\right\rangle = O\left(\frac{1}{\di^4}\right).
\end{equation}
Indeed a simple change of variables gives
\begin{equation*}
\left\langle\Phi, B_{\di+h} \Psi\right\rangle = \left\langle\gamma_{h}\Phi, B_{\di} \gamma_{h}\Psi\right\rangle,
\end{equation*}
where $\gamma_{h}$ is defined by
\begin{equation*}\label{def:delta}
	\gamma_{h}\psix(x_1,\dots,x_{N_1+N_2}) = \psix(x_1+he_1,\dots,x_{N_1+N_2}+he_1).
\end{equation*}
Thus
\begin{align}
\lim_{h \rightarrow 0} \left\langle\Phi, \frac{B_{\di+h}-B_{\di}}{h} \Psi\right\rangle &= \lim_{h \rightarrow 0}\left\langle\frac{\gamma_{h}\Phi-\Phi}{h}, B_{\di} \gamma_{h}\Psi\right\rangle + \lim_{h \rightarrow 0}\left\langle\Phi, B_{\di} \frac{\gamma_{h}\Psi-\Psi}{h}\right\rangle \nonumber
\\ &= \left\langle(-\vec{w}\cdot\nabla)\Phi, B_{\di} \Psi\right\rangle + \left\langle\Phi, B_{\di} (-\vec{w}\cdot\nabla)\Psi\right\rangle \label{eq:limBr}
\end{align}
where 
\begin{equation*}
\vec{w}= \underbrace{(-1,0,0,-1,0,0,\dots,-1,0,0)}_{N_1+N_2 \text{ times}},
\end{equation*}
and we used that the momentum operator is the generator of translations. Since $\|\Psi\|_{H^2}=O(1)$ and $\|\Phi\|_{H^2} = O\left(\frac{1}{\di^4}\right)$ we can apply Hardy's inequality to the right hand side of \eqref{eq:limBr} to arrive at \eqref{est:limBr}. From \eqref{eq:Mder2}, \eqref{eq:Hxder} and \eqref{est:limBr} we obtain
\begin{equation*}\label{est:Mder2}
\frac{\partial^2 M}{\partial\di^2} = O\left(\frac{1}{\di^4}\right).
\end{equation*}
Together with  \eqref{est:M} and \eqref{est:Mder} we arrive at \eqref{est:Mderalpha}.
Using \eqref{def:B} and \eqref{est:Mderalpha} we find that
\begin{equation}
\label{est:Bderalpha}
\frac{\partial^\alpha B}{\partial\di^\alpha} = O\left(\frac{1}{\di^{10}}\right) \qquad\text{for }\alpha\in\{0,1,2\}.
\end{equation}

To estimate the term $C$ defined in \eqref{def:C} we will introduce a new cut-off function.
Let
\begin{equation}\label{def:Chix}
\Chix(x_1,\dots,x_{N_1+N_2}) = \prod_{i=1}^{N_1}\chi_1\left(\left(\frac{6}{7}\right)^2\frac{x_i}{\di_0}\right)\prod_{j=N_1+1}^{N_1+N_2}\chi_1\left(\left(\frac{6}{7}\right)^2\left(\frac{x_j}{\di_0}-\frac{\di}{\di_0}e_1\right)\right).
\end{equation}
This dilation ensures that
\begin{equation}\label{eq:suppChix}
\left\lbrace
	\begin{aligned}
	&\Chix = 1 \text{ on } \supp(\psix),\\
	&\Chix T_\pi\psix = 0 \text{ for all } \pi \in S_{N_1+N_2}\backslash \subi \times \subj.
	\end{aligned}
\right.	
\end{equation}
From \eqref{def:C}, \eqref{eq:normQpsi} and \eqref{eq:suppChix} we have
\begin{equation*}
C = \binom{N_1+N_2}{N_1} \langle\ (\Hx-E(\infty))\psix, \Chix \left(\Hxabot-E\right)^{-1}(\Hx-E(\infty))\cQ\psix\rangle.
\end{equation*}
Since
\begin{equation*}
\Chix\left(\Hxabot-E\right)^{-1} = \Rx\Chix + \left[\Chix \left(\Hxabot-E\right)^{-1} - \Rx\Chix\right],
\end{equation*}
with $\Rx$ defined in \eqref{def:Rr} and by \eqref{eq:suppChix} and the locality of $H$
\begin{equation*}
\Chix(\Hx-E(\infty))\cQ\psix = \frac{1}{\binom{N_1+N_2}{N_1}}(\Hx-E(\infty))\psix,
\end{equation*}
We find
\begin{equation}\label{eq:C1C2}
C = C_1 +C_2,
\end{equation}
where
\begin{equation*}\label{def:C1}
C_1 = \left\langle (\Hx-E(\infty))\psix, \Rx(\Hx-E(\infty))\psix \right\rangle
\end{equation*}
and
\begin{equation*}\label{def:C2}
C_2 = \left\langle (\Hx-E(\infty))\psix, \left[\Chix \left(\Hxabot-E\right)^{-1} - \Rx\Chix\right] (\Hx-E(\infty))\cQ\psix \right\rangle.
\end{equation*}
By the same change of variables as in section \ref{sec:lin} we have 
\begin{equation*}
C_1 = \left\langle (\Hz-E(\infty))\psiz, \Rz(\Hz-E(\infty))\psiz \right\rangle,
\end{equation*}
and therefore by \eqref{est:HzE0} and the $\di$-independence of $\Rz$, when $\di$ varies a little bit, we find that
\begin{equation}
\label{est:C1}
C_1 = \Odi\left(\frac{1}{\di^8}\right).
\end{equation}
To estimate $C_2$ observe that \eqref{def:Rr} and \eqref{def:Px0} imply that
\begin{align}
&\left[\Chix \left(\Hxabot-E\right)^{-1} - \Rx\Chix\right]\cQ\nn\\
&\  = \Rx \left[(\Hxabot_0-E(\infty))\Chix - \Chix(\Hxabot-E)  \right]\cQ\left(\Hxabot-E\right)^{-1}\nn\\
&\   = \Rx\left[\Pxbot_0\Hx_0\Pxbot_0\Chix - \Chix\Pxbot\Hx\Pxbot  + (E-E(\infty))\Chix\right]\cQ\left(\Hxabot-E\right)^{-1},\label{eq:diffResolR}
\end{align}
where in the last step we used that the projection $\cQ$ commutes with $\Pxbot,\Hx^\bot$, and that $\cQ_1,\cQ_2$ commute with $\Pxbot_0,\Hx_0$ and the equality $\cQ_1\cQ_2\cQ = \cQ$ holds, to omit the $\cQ_1,\cQ_2,\cQ$ that appear in the definitions of $\Hx_0^{a,\bot},\Hxabot$, respectively.
We now claim that
\begin{equation}\label{eq:XPbotPbotpsi}
\Chix\Pxbot\Hx\Pxbot\cQ = \Pxbot_\psix\Chix\Hx\Pxbot_\psix\cQ,
\end{equation}
where $\Pxbot_\psix = 1 - |\psix\rangle\langle\psix|$. Indeed from \eqref{eq:suppChix} it follows that $\Chix\cQ\psix = \Chix\psix\|\cQ\psix\|^2$, which together with \eqref{eq:projP} and the fact that the projection $\cQ$ commutes with $\Hx$ and $\Pxbot$ gives
\begin{equation*}\label{preeq:XPbotPbotpsi}
\Chix\Pxbot\Hx\Pxbot\cQ = \Chix\Pxbot_\psix\Hx\Pxbot\cQ = \Pxbot_\psix\Chix\Hx\Pxbot\cQ
\end{equation*}
where in the last step we used that $\Chix\psix = \psix$, see \eqref{eq:suppChix}. Due to \eqref{eq:suppChix} and the locality of $H$ we have that $\Chix\Hx\cQ\psix = \Chix\Hx\psix\|\cQ\psix\|^2$ and thus we can repeat the argument  to arrive at \eqref{eq:XPbotPbotpsi}.

Due to \eqref{eq:suppChix} $\Chix$ commutes with $\Px_0$ and thus
\begin{equation*}
\label{eq:XP0comm}
\Pxbot_0\Hx_0\Pxbot_0\Chix = \Pxbot_0\Hx_0\Chix\Pxbot_0.
\end{equation*}
Together with  \eqref{eq:diffResolR} and \eqref{eq:XPbotPbotpsi} we find
\begin{multline*}
\left[\Chix \left(\Hxabot-E\right)^{-1} - \Rx\Chix\right]\cQ  
\\ = \Rx\left[\Pxbot_0\Hx_0\Chix\Pxbot_0 - \Pxbot_\psix\Chix\Hx\Pxbot_\psix + (E-E(\infty))\Chix\right]\left(\Hxabot-E\right)^{-1} \cQ.
\end{multline*}
Observe that
\begin{multline}\label{eq:Px0Pxpsix}
\Pxbot_0\Hx_0\Chix\Pxbot_0 - \Pxbot_\psix\Chix\Hx\Pxbot_\psix =(\Px_\psix-\Px_0)\Hx_0\Chix\Pxbot_0\\ + \Pxbot_\psix (\Hx_0\Chix - \Chix\Hx)\Pxbot_0 +\Pxbot_\psix\Chix\Hx(\Px_\psix -\Px_0).
\end{multline}
Thus, we obtain
\begin{equation*}
\label{def:C21C22C23C24}
C_2 = C_{21} + C_{22} + C_{23} + C_{24},
\end{equation*}
where
\begin{align*}
	C_{21} &= \langle(\Hx-E(\infty))\psix, \Rx (E-E(\infty))\Chix (\Hxabot-E)^{-1} (\Hx-E(\infty))\cQ\psix\rangle,
\\  C_{22} &= \langle(\Hx-E(\infty))\psix, \Rx \left((\Px_\psix-\Px_0)\Hx_0\Chix\Pxbot_0\right) \cQ(\Hxabot-E)^{-1} (\Hx-E(\infty))\cQ\psix\rangle,
\\	C_{23} &= \langle(\Hx-E(\infty))\psix, \Rx \left(\Pxbot_\psix\Chix\Hx(\Px_\psix -\Px_0)\right)\cQ (\Hxabot-E)^{-1} (\Hx-E(\infty))\cQ\psix\rangle,
\\	C_{24} &= \langle(\Hx-E(\infty))\psix, \Rx \left(\Pxbot_\psix (\Hx_0\Chix - \Chix\Hx)\Pxbot_0\right) (\Hxabot-E)^{-1} (\Hx-E(\infty))\cQ\psix\rangle.
\end{align*}
We first estimate $C_{21}$.
By \eqref{eq:E-Er} and \eqref{eq:old} we have
\begin{equation*}\label{est:E-Einf}
\frac{\partial^\alpha}{\partial r^\alpha}(E - E(\infty)) = O\left(\frac{1}{\di^6}\right)
\end{equation*}
for $\alpha\in\{0,1,2\}$.
Furthermore one can verify that $\partial_r^\alpha\Chix = O(r^{-\alpha})$ for $\alpha\in\{0,1,2\}$.
Using \eqref{est:QHxpsixderalpha} and \eqref{est:QHxpsixder2} we can argue similarly as in the proof of \eqref{est:Mderalpha} to arrive at
\begin{equation}
\label{est:C21deralpha}
\frac{\partial^\alpha}{\partial\di^\alpha}C_{21} = O\left(\frac{1}{\di^{14}}\right) \qquad\text{for }\alpha\in\{0,1,2\}.
\end{equation}
Note that differentiating the term having $\frac{d\Rx}{d\di}$ can be similarly handled as differentiating the term containing $\frac{\partial}{\partial \di}(\Hxabot-E)^{-1}$ in the proof of \eqref{est:Mderalpha}.

We next estimate $C_{24}$. Observe that
\begin{equation*}
\Hx_0\Chix - \Chix\Hx = -\Ix\Chix - \left[\Delta,\Chix\right] = -\Ix\Chix -(\Delta\Chix) -2\nabla\Chix\cdot\nabla,
\end{equation*} 
and thus by the previous estimate on $\partial_r^\alpha\Chix$ and the fact that $X$ is supported $O(\di)$ far from the singularities of $I$ we have 
\begin{equation*}
\frac{d^\alpha}{d\di^\alpha}\left(\Hx_0\Chix - \Chix\Hx\right) = O\left(\frac{1}{\di^{\alpha+1}}\right) +O\left(\frac{1}{\di^{\alpha+1}}\right) \cdot\nabla
\end{equation*}
for $\alpha\in\{0,1,2\}$.
Using the $(\Hxabot-E)$ boundedness of $ Q \nabla$ we can argue again as in the proof of \eqref{est:C21deralpha} to conclude that
\begin{equation}
\label{est:C24deralpha}
\frac{\partial^\alpha C_{24}}{\partial\di^\alpha} = O\left(\frac{1}{\di^9}\right) \qquad \text{for }\alpha\in\{0,1,2\}.
\end{equation}

We now estimate $C_{22}$. Observe that
\begin{equation}\label{eq:diffPx0Pxpsix}
\Px_\psix-\Px_0 = |\fr\rangle\langle\psix| + |\ppy\rangle\langle \fr|,
\end{equation}
where, by~\eqref{def:psi0}, 
\begin{align*}
\fr :=\psix-\ppy &= \left(\frac{1}{\|\psiz_0\|^2}-1\right)\ppy + \frac{\chix\Rx\Ix\ppy}{\|\psiz_0\|^2}
\\ &= \tau_\di\left[\left(\frac{1}{\|\psiz_0\|^2}-1\right)\pp + \frac{\chiz\Rz\Iz\pp}{\|\psiz_0\|^2}\right].
\end{align*}
It follows from Lemma \ref{lem:Mrder} together with \eqref{Ifnocutoff} and \eqref{est:psi0norm} that
\begin{equation*}
\frac{d^\alpha}{d\di^\alpha}\fr = O\left(\frac{1}{\di^3}\right)
\end{equation*}
for $\alpha\in\{0,1,2\}$ and thus by \eqref{eq:diffPx0Pxpsix}
\begin{equation}
\label{est:Px0Pxpsixderalpha}
\frac{d^\alpha}{d\di^\alpha}(\Px_\psix-\Px_0) = O\left(\frac{1}{\di^3}\right), \quad \alpha\in\{0,1,2\}. 
\end{equation}
Since $\Chix=1$ on $\supp\psix_0$ and $\supp\ppy$ we have
\begin{equation*}
C_{22} = \langle(\Hx-E(\infty))\psix, \Rx \left((\Px_\psix-\Px_0)\Hx_0\Pxbot_0\right) \cQ(\Hxabot-E)^{-1} (\Hx-E(\infty))\cQ\psix\rangle.
\end{equation*}
Thus using \eqref{eq:diffPx0Pxpsix}
\begin{equation*}
\label{def:C221C222}
C_{22} = C_{221} + C_{222},
\end{equation*}
where
\begin{equation*}\label{def:C221}
C_{221} = \langle(\Hx-E(\infty))\psix, \Rx \fr\rangle \langle \psix,\Hx_0\Pxbot_0\cQ (\Hxabot-E)^{-1} (\Hx-E(\infty))\cQ\psix\rangle
\end{equation*}
and
\begin{equation*}\label{def:C222}
C_{222} = \langle(\Hx-E(\infty))\psix, \Rx\ppy\rangle \langle\fr,\Hx_0\Pxbot_0\cQ (\Hxabot-E)^{-1} (\Hx-E(\infty))\cQ\psix\rangle.
\end{equation*}
We now estimate $C_{221}$. Using the change of variables of Section \ref{sec:lin} we find
\begin{equation*}
C_{221} = \langle(\Hz-E(\infty))\psiz, \Rz \frz\rangle \langle \psix,\Hx_0\Pxbot_0\cQ (\Hxabot-E)^{-1} (\Hx-E(\infty))\cQ\psix\rangle.
\end{equation*}
But by~\eqref{Ifnocutoff}
\begin{equation*}
\frz = \left(\frac{1}{\|\psiz_0\|^2}-1\right)\pp + \frac{\chiz\Rz\Iz\pp}{\|\psiz_0\|^2} = \Od^\infty\left(\frac{1}{\di^3}\right)
\end{equation*}
and it follows using \eqref{est:HzE0} that
\begin{equation*}
C_{221} = \Odi\left(\frac{1}{\di^{7}}\right)\langle \psix,\Hx_0\Pxbot_0\cQ (\Hxabot-E)^{-1} (\Hx-E(\infty))\cQ\psix\rangle.
\end{equation*}
Because of the boundedness of $\Hx_0\Pxbot_0\cQ (\Hxabot-E)^{-1}$ we can differentiate $\psix$ twice.
Note that in $\frac{\partial}{\partial\di}C_{221}$ the term containing $\frac{d\Hx_0}{d\di}$ appears. To differentiate this term we argue as in the proof of \eqref{est:Mder2}.
With these additional observations we can argue as in the proof of \eqref{est:C21deralpha} to find that $\partial_r^\alpha C_{221}=O(r^{-11})$ for all $\alpha\in\{0,1,2\}$. Similarly we can show that $\partial_r^\alpha C_{222}=O(r^{-11})$
for all $\alpha\in\{0,1,2\}$. Thus we arrive at $\partial_r^\alpha C_{22}=O(r^{-11})$ and $\partial_r^\alpha C_{23}=O(r^{-11})$ and obtain $\partial_r^\alpha C_{2}=O(r^{-9})$, for all $\alpha\in\{0,1,2\}$.
Together with \eqref{eq:C1C2} and \eqref{est:C1} we conclude that 
\begin{equation}
\label{est:Cderalpha}
C = O\left(\frac{1}{\di^8}\right),\qquad \frac{\partial^\alpha}{\partial\di^\alpha}C = O\left(\frac{1}{\di^9}\right)
\end{equation}
for $\alpha\in\{1,2\}$ and hence arrive at \eqref{est:Aalpha}.

\subsection{Conclusion of the proof of Theorem \ref{thm:new}}
We now apply the implicit function theorem in \eqref{eq:ImprE} to estimate the derivatives of $E(r)$, which coincide with that of $W(r)$. From \eqref{eq:ImprE}, \eqref{eq:psiHpsi}, \eqref{est:Aalpha}, we find
\begin{equation*}\label{impl1}
\frac{\partial \Imp(\di,E)}{\partial r}=-\frac{6 \sigma}{\di^7} + O_d^1 \left(\frac{1}{\di^8}\right)
\end{equation*}
for $E$ as in \eqref{eq:E-Er}.
Since $\frac{\partial}{\partial E} \left(\Hx^{a,\bot}-E\right)^{-1}=\left(\Hx^{a,\bot}-E\right)^{-2}$, using \eqref{eq:ImprE}, \eqref{eq:nonlin}, \eqref{est:HxEpsix} we obtain that
\begin{equation*}\label{impl2}
\frac{\partial \Imp(\di,E)}{\partial E}=1 + O\left(\frac{1}{\di^8}\right) \neq 0.
\end{equation*}
for $E$ close to $E(r)$. Thus $\Imp$ has continuous partial derivatives of first order and is continuously differentiable in a neighborhood of the curve $(r,E(r)), r\in (r_0-\delta,r_0+\delta)$ for some $\delta>0$. Since  \eqref{eq:Erexp} also holds we can apply the implicit function theorem to conclude that if $r$ is large enough then $E(r)$ is differentiable and 
\begin{equation*}\label{WrAb}
W'(r)=E'(r)=-\frac{\partial_r \Imp(\di,E(r))}{\partial_E \Imp(\di,E(r))}.
\end{equation*}
We can conclude the proof of the estimate on $W'(r)$ in Theorem \ref{thm:new}. In fact with the argument providing Equation \eqref{impl1} it follows that 
\begin{equation*}\label{impl3}
\frac{\partial^2 \Imp(\di,E(\di))}{\partial \di \partial E}= O_d^1\left(\frac{1}{\di^8}\right).
\end{equation*}
and since $\frac{\partial}{\partial E} \left(\Hx^{a,\bot}-E\right)^{-2}=2 \left(\Hx^{a,\bot}-E\right)^{-3}$ we  obtain that
\begin{equation*}\label{impl4}
\frac{\partial^2 \Imp(\di,E(\di))}{\partial E^2}= O_d^1\left(\frac{1}{\di^8}\right).
\end{equation*}
Hence we can differentiate the right hand side of 
\eqref{WrAb} and we obtain the estimate on $W''(r)$ stated in Theorem \ref{thm:new}.\qed

\section{Proof  of Theorem \ref{thm:new2}}

In this section we provide a proof of Theorem \ref{thm:new2}. To handle the spin we need to introduce further notations. 
Let $G_k$ be the ground state eigenspace of $\Hx_k^a$ and define the approximate cut-off ground state eigenspace 
\begin{equation*}
G_{k,r}:=\left\{\prod_{i=1}^{N_j}\chi_1\left(\frac{x_i}{\di_0}\right) \phi: \phi \in G_k\right\},
\end{equation*}
where $\di_0$ is the same as in \eqref{def:phi}. Let $\{\varphi_n:n=1,\dots,\dim(G_1)\}$ be an orthonormal basis of $G_{1,r}$ and 
$\{\widetilde{\psi_m}:m=1,\dots,\dim(G_2)\}$ be an orthonormal basis of $G_{2,r}$.  
We further define
\begin{equation}
\label{def:Psinm}
\Psinm{}{} := \phipsi{}{} - \chix \Rx \Ix \phipsi{}{} = K \phipsi{}{},
\end{equation}
where $K = (1- \chix\Rx\Ix)$ and 
\begin{equation*}
\psi_m(x_1,\dots,x_{N_2})=\widetilde{\psi_m}(x_1-re_1,\dots,x_{N_2}-re_1) =: (\tau_r\widetilde{\psi_m})(x_1,\dots,x_{N_2}),
\end{equation*}
with $R$ is defined as in \eqref{def:Rr} with the difference that $\Px_0$ is the orthogonal projection onto $G_{1,\di}\otimes \tau_r G_{2,\di}$.
Finally, we define $\Pi$ as the orthogonal projection onto $$\text{span}\{\cQ\Psinm{}{}|\quad n=1,\dots,\dim(G_1),\quad m=1,\dots,\dim(G_2)\}.$$

We assume without loss of generality that the second atom has an irreducible ground state eigenspace. We will prove the following  lemma, which will help us to adapt the arguments of Section \ref{sec:proof} in the present setting. To this end we define the operator $\Pi_2$ to be the orthogonal projection onto $ \tau_r G_{2,\di} $. Recall that $Y$ is defined in \eqref{def:YN1N2}.
\begin{lemma}\label{lem:PiHPi}
	Let $A$ be a self-adjoint operator on $Y$ acting only on the position variables and whose domain contains $ G_{1,r} \otimes \tau_r G_{2,r}$. Then the operator $S: \tau_r G_{2,\di} \to \tau_r G_{2,\di} $ defined through the sesquilinear form
	$$\langle \psi_{m_1}, S \psi_{m_2} \rangle= \langle  \Phi_1 \otimes \psi_{m_1}, A  \Phi_2 \otimes \psi_{m_2} \rangle, $$
	with $\Phi_1,\Phi_2 \in G_{1,\di}$, is a multiple of the identity. The same statement holds if  $\tau_r G_{2,r}$ is replaced by $G_{2,r}$. 
\end{lemma}
\begin{proof}
	The operator $S$ is explicitly given by
	 $S \psi=\langle  \Phi_{1}, \Pi_2 A \Pi_2 \Phi_{2} \otimes \psi \rangle_1 $,
	 where $\langle\cdot \rangle_1$ indicates that the integration for the inner product is taken only with respect to the coordinates of the first atom. One can  verify that $S^* \psi= \langle  \Phi_{2}, \Pi_2 A \Pi_2 \Phi_{1} \otimes \psi \rangle_1$. 
	 We write now $S$ as a linear combination of two self-adjoint operators on $\tau_r G_{2,r}$ namely  
	 \begin{equation*}
		S = \frac{S + S^*}{2}-i \frac{iS -i S^*}{2} =: B+Ci.
	 \end{equation*}
	 Since $A$ acts only on the position variables, it follows that
 the two self-adjoint  operators $B,C$ commute with $\widetilde{T}_\pi$ for all permutations $\pi \in S_{N_2}$ and with the spin shift $\cF$. Thus, due to the irreducibility of the ground state eigenspace of the second atom, the operators $B$ and $C$  have to be  multiples of the identity on $\tau_r G_{2,r}$. Indeed, since e.g. $B$ is symmetric and bounded on $\tau_r G_{2,r}$, there exists an eigenvector $f\in \tau_r G_{2,r}$ so that
$Bf = \lam f$. Since moreover $B$ commutes with $\Tspin$ and $\cF$, we find that $\ker(B-\lambda)$ is invariant under the group of spin transformations. The irreducibility implies that $\ker(B-\lambda)=G_{2,r}$, that is, $B$ is a multiple of the identity. The same argument for $C$ gives that $S$ is as well a (complex) multiple of the identity. 
\end{proof}

Arguing as  \cite{AS}, see also \cite{MoS}, it follows that there exist $c,C$ such that if  $\di\geq c$ then
\begin{equation}\label{eq:IMScor2}
	\Pi^\bot\Hx^a\Pi^\bot - E(\di) \geq C > 0.
\end{equation}
Thus we can apply as in the spinless case the Feshbach map
\begin{equation*}
	F_\Pi(\lam) = \left.\big(\Pi\Hx\Pi - \Pi\Hx\Pi^\bot(\Hxabot -\lam)^{-1}\Pi^\bot\Hx\Pi\big)\right|_{\Ran \Pi},
\end{equation*}
where $\Hxabot = \Pi^\bot\Hx^{a}\Pi^\bot$ and it follows that $E=E(r)$ is an eigenvalue of $F_\Pi(E)$. Moreover, arguing as in the proof of \eqref{est:HzE0} we can show that
\begin{equation*}
	\left\| \Pi^\bot\Hx\Pi \right\| \leq O\left(\frac{1}{\di^4}\right).
\end{equation*}
It therefore follows that 
\begin{equation*}
	E = \min \sigma\left(\left.\Pi\Hx\Pi\right|_{\Ran \Pi}\right)+ O \left(\frac{1}{\di^8}\right).
\end{equation*}
Let $\Psi\in\Ran \Pi$ be a minimizer of $\langle\Phi,\Pi\Hx\Pi\Phi\rangle$, $\Phi\in\Ran \Pi$, $\|\Phi\|=1$. Then
\begin{equation*}
	E = \langle\Psi,\Hx\Psi\rangle + O\left(\frac{1}{\di^8}\right),
\end{equation*}
and
\begin{equation}\label{eq:defpsipsi0}
	\Psi = \frac{\cQ\Psi_0}{\|\cQ\Psi_0\|},\qquad \text{ where } \Psi_0 = \sum_{m=1}^{\text{dim}(G_2)} K\Phi_m\otimes\psi_{m}
\end{equation}
 for some $\Phi_m\in G_{1,\di}$ such that $\sum_{m=1}^{\dim(G_2)}\|\Phi_m\|^2=1$. So arguing as in \eqref{eq:PHP} we find
\begin{equation*}
	E - \Einfty  = \frac{1}{\|\Psi_0\|^2}\langle\Psi_0,(\Hx - \Einfty)\Psi_0\rangle +O\left(\frac{1}{r^8}\right),
\end{equation*}
which together with Lemma \ref{lem:PiHPi} for $A=K^*(H-E(\infty)) K$ gives that
\begin{equation}\label{eq:interensum}
E - \Einfty  = \frac{1}{\|\Psi_0\|^2} \sum_{m=1}^{\dim(G_2)} \langle \Phi_m\otimes\psi_{m}, K^*\left(\Hx-\Einfty\right)K\Phi_m\otimes\psi_{m}\rangle  + O\left(\frac{1}{r^8}\right).
\end{equation}
We are going to prove that
\begin{multline}\label{eq:prevdW}
\left\langle \Phi_m\otimes\psi_{m}, K^*\left(\Hx-\Einfty\right)K\Phi_m\otimes\psi_{m}\right\rangle \\=  -\frac{\langle f\Phi_m\otimes\widetilde{\psi_{m}}, \Rz f \Phi_m\otimes\widetilde{\psi_{m}}\rangle}{\di^6} +O\left(\frac{1}{\di^7}\right).
\end{multline}
This can be done as in the proof of \eqref{eq:psiHpsi} with the help of \eqref{eq:N1N2N3}, \eqref{eq:PHPN1}, \eqref{eq:PHPN2}, \eqref{eq:PHPN3}. The only things which are not a priori clear are that
\begin{equation}\label{eq:IPhipsi}
\langle \Phi_m\otimes\widetilde{\psi_n}, \Iz \Phi_m\otimes\widetilde{\psi_n}\rangle = 0,
\end{equation}
and that
\begin{equation}\label{eq:P0botIPhipsi}
\Pzbot_{G_{1,r} \otimes G_{2,r}}\Iz\Phi_m\otimes\widetilde{\psi_n} = \Iz\Phi_m\otimes\widetilde{\psi_n}.
\end{equation}
We will now prove \eqref{eq:IPhipsi} and \eqref{eq:P0botIPhipsi} can be similarly proven.
Using Lemma \ref{lem:PiHPi} we find 
\begin{equation}\label{eq:Newt1}
	\langle \Phi_m\otimes \widetilde{\psi_{m}}, \Iz \Phi_m\otimes\widetilde{\psi_{m}}\rangle = \frac{1}{\dim(G_2)} \sum_{n=1}^{\dim(G_2)}\langle \Phi_m\otimes\widetilde{\psi_n}, \Iz \Phi_m\otimes\widetilde{\psi_n}\rangle.
\end{equation}
We will now prove that
\begin{equation}\label{eq:Newt2}  \sum_{n=1}^{\dim(G_2)}\langle \Phi_m\otimes\widetilde{\psi_n}, \Iz \Phi_m\otimes\widetilde{\psi_n}\rangle=0.
\end{equation}
Indeed we have that
\begin{equation}\label{eq:Newt3} 
\sum_{n=1}^{\dim(G_2)}\langle \Phi_m \otimes\widetilde{\psi_n}, \Iz \Phi_m \otimes\widetilde{\psi_n}\rangle = \langle\Phi_m,V\Phi_m\rangle ,
\end{equation}
where 
\begin{align*}
&V\left(x_1,\dots,x_{N_1}\right) \\ 
&= \sum_{n=1}^{\dim(G_2)} \int\Iz\left(x_1,\dots,x_{N_1+N_2}\right)|\widetilde{\psi_n}|^2\left(x_{N_1+1},\dots,x_{N_1+N_2}\right)dx_{N_1+1}\dots 
\\ &= \sum_{n=1}^{\dim(G_2)} \int\sum_{j=1}^{N_2}\Ix_0\left(x_1,\dots,x_{N_1},x_{N_1+j}\right)|\widetilde{\psi_n}|^2\left(x_{N_1+1},\dots,x_{N_1+N_2}\right)dx_{N_1+1}\dots ,
\end{align*}
with
\begin{equation*}
\Ix_0(x_1,\dots,x_{N_1},x) = \sum_{i=1}^{N_1}\left(\frac{1}{|re_1 + x-x_i|}-\frac{1}{|\di e_1 -x_i|} \right) -\frac{N_1}{|re_1 + x|} + \frac{N_1}{\di}.
\end{equation*}
Therefore, we find
\begin{equation*}
V(x_1,\dots,x_{N_1}) = \int\Ix_0\left(x_1,\dots,x_{N_1},x\right)\rho(x)dx,
\end{equation*}
where $\rho$ is the electron density of the second atom given by
\begin{equation*}\label{def:elden}
\rho(x) = \frac{1}{\dim(G_2)}\sum_{n=1}^{\dim(G_2)}\sum_{j=1}^{N_2}\int|\widetilde{\psi_n}|^2\left(x_{1},\dots,x_{j-1},x,x_{j+1},\dots,x_{N_2}\right)
\end{equation*}
where all the variables are integrated with the exception of $x$.
Note that $\widetilde{\psi_n}$ is an orthonormal basis of $G_{2,r}$. It is well known that $\rho(x)$ is spherically symmetric, see for example \cite{AS}. Thus applying Newton's theorem we find that $V(x_1,\dots,x_{N_1})$ vanishes on the support of $\Phi_m$, which together with \eqref{eq:Newt3} implies \eqref{eq:Newt2}. From \eqref{eq:Newt1} and \eqref{eq:Newt2}
we arrive at \eqref{eq:IPhipsi} and \eqref{eq:P0botIPhipsi} can be similarly proven. Thus \eqref{eq:prevdW} holds.
From \eqref{eq:prevdW}
we find
\begin{multline*}
	\sum_{m=1}^{\dim(G_2)}\langle \Phi_m\otimes\psi_{m}, K^*\left(\Hx-\Einfty\right)K\Phi_m\otimes\psi_{m}\rangle \\
	= -\sum_{m=1}^{\dim(G_2)} \frac{\langle f\Phi_m\otimes\psi_{m}, \Rx f \Phi_m\otimes\psi_{m}\rangle}{\di^6} + O\left(\frac{1}{\di^7}\right).
\end{multline*}
Since $\Psi$ given by \eqref{eq:defpsipsi0} minimizes the quadratic form of $\langle\Phi,\Hx\Phi\rangle$ for $\|\Phi\|= 1, \Phi\in\Ran \Pi$, we obtain with Lemma \ref{lem:PiHPi} and \eqref{sigmaijvdef} that
\begin{equation}\label{eq:fastvdW}
	\sum_{m=1}^{\dim(G_2)}\langle \Phi_m\otimes\psi_{m}, K^*\left(\Hx-\Einfty\right)K\Phi_m\otimes\psi_{m}\rangle 
 	 = -\frac{\sigma}{\di^6} + O\left(\frac{1}{\di^7}\right).
\end{equation}
With the help of Lemma \ref{lem:PiHPi} we see that $\|\Psi_0\|^2=\sum_{m=1}^{\dim(G_2)} \|K \Phi_m \otimes \psi_m\|^2$, as the the functions on the right hand side are orthogonal to each other. Thus since $\sum_{m=1}^{\dim(G_2)}\|\Phi_m\|^2=1$ we can prove  similarly as \eqref{est:psi0norm} that
\begin{equation}\label{eq:Psionorm} 
\|\Psi_0\|^2=1+O_d^\infty\left(\frac{1}{r^6}\right).
\end{equation}
Using \eqref{eq:interensum}, \eqref{eq:fastvdW} and \eqref{eq:Psionorm}, we arrive at \eqref{eq:new2}. 

It thus remains to show that $W(r)$ or equivalently $E(r)$ is strictly increasing for large $r$, which we do next.

\begin{lemma}[Monotonicity]
There exists a $d>0$ such that $E(\di)$ is strictly increasing on $[d,\infty)$.
\end{lemma}
\begin{proof} 
	Here we adapt ideas of Section 3 in \cite{AL}. There the situation was  different, a system of two molecules that can be rotated was considered, but it was assumed that both of them have irreducible ground state eigenspaces.
	Let $\Psi_{s}=\frac{Q\Psi_0}{\|Q\Psi_0\|}$, be an eigenvector to the eigenvalue $E(s)$ of the Feshbach map $F_\Pi(E(s))$ for distance $s$, with $\Psi_0$ defined  similarly as in \eqref{eq:defpsipsi0}. Let
	\begin{equation*}
	\Psi_{\di} = \frac{Q\tau_{\di-s}\Psi_0}{\|Q\tau_{\di-s}\Psi_0\|},
	\end{equation*}
	where  $\tau$ is given by \eqref{def:tau} and define
\begin{equation}\label{eq:Drdec}
D(\di) := \langle\Psi_{\di}, \Hx_{\di}\Psi_{\di}\rangle - N(\di),
\end{equation}
where
\begin{equation*}
N(\di) := \langle\Psi_{\di}, \Hx_{\di} \Pi_r^\bot(\Hxabot_{\di}-E(s))^{-1} \Pi_r^\bot\Hx_{\di}\Psi_{\di}\rangle
\end{equation*}
and we wrote $\Pi_\di$ to emphasize the dependence of the projection $\Pi$ on $\di$.  
Note that $\Psi_\di$ is in the range of the projection $\Pi_\di$ for the Feshbach map but it is not necessarily eigenfunction to the eigenvalue $E(\di)$ of the Feshbach map $F_{\Pi_{\di}}(\di)$. However, it can be used as a test function. We also emphasize the fact that $E(s)$ does not change in the definition of $N(\di)$ when $\di$ changes, which helps a lot the analysis. Arguing as in \eqref{eq:PHP} and \eqref{eq:PHE0P} we find 
that
\begin{equation*}
\langle\Psi_{\di}, \Hx_{\di}\Psi_{\di}\rangle=\frac{1}{\|\Psi_0\|^2}\langle\widetilde{\Psi_0}, \widetilde{\Hx_{\di}} \widetilde{\Psi_0} \rangle,
\end{equation*}
which with the help of \eqref{eq:defpsipsi0} and Lemma  \ref{lem:PiHPi} 
becomes 
\begin{equation*}
\langle\Psi_{\di}, \Hx_{\di}\Psi_{\di}\rangle=\frac{1}{\|\Psi_0\|^2}\sum_{m=1}^{\dim(G_2)}\langle \Phi_m \otimes \widetilde{\psi_m}, \widetilde{K^*} \widetilde{\Hx_{\di}}
 \widetilde{K}  \Phi_m \otimes \widetilde{\psi_m} \rangle,
\end{equation*}
where $\widetilde{K}=1-\widetilde{\chi} \widetilde{R} \widetilde{I_r}$.
Thus, since $\sum_{m=1}^{\dim(G_2)} \|\Phi_m\|^2=1$ and \eqref{eq:Psionorm} holds, 
we  can argue as in the proof of \eqref{eq:psiHpsi} to show that
\begin{equation}\label{eq:psiHpsideg}
\langle\Psi_{\di}, \Hx_{\di}\Psi_{\di}\rangle = \Einfty -\frac{\sigma}{\di^6}+\Odi\left(\frac{1}{\di^7}\right).
\end{equation}
To estimate $N(r)$ we observe that 
\begin{equation*}\label{eq:Nr}
N(\di) = \langle \Pi_r^\bot (\Hx_{\di}-E(\infty)) \Psi_{\di},  (\Hxabot_{\di}-E(s))^{-1} \Pi_r^\bot(\Hx_{\di}-E(\infty))\Psi_{\di}\rangle.
\end{equation*}
Since by \eqref{eq:defpsipsi0} it follows that 
$$(\Hx_{\di}-E(\infty)) \Psi_{\di}=\sum_{m=1}^{\dim(G_2)} \frac{Q}{\|Q \Psi_0\|} (H-E(\infty)) K \Phi_m \otimes \psi_m,$$ 
using 
\eqref{eq:Psionorm}  and arguing for each of the summands as in the proof of \eqref{est:QHxpsixderalpha} we find
\begin{equation*}\label{eq:HminusEpsider}
\frac{d^a}{dr^a} (\Hx_{\di}-E(\infty)) \Psi_{\di}=O\left(\frac{1}{\di^4}\right)
\end{equation*}
for $a=0,1$.
Arguing as in the proofs of \eqref{est:Qpsix} and \eqref{est:Resolder} we find, respectively, that
\begin{equation*}\label{eq:Pirder}
\frac{d}{dr}\Pi_\di^\bot=O(1)
\end{equation*}
and 
\begin{equation*}\label{eq:resder}
\left[\frac{d}{dr}(\Hxabot_{\di}-E(s))^{-1}\right] \cQ= O(1).
\end{equation*}
We arrive at
\begin{equation}\label{eq:derN}
\frac{d^a}{dr^a}N= O\left(\frac{1}{\di^8}\right)
\end{equation}
for $a=0,1$.
	It thus   follows from \eqref{eq:Drdec}, \eqref{eq:psiHpsideg} and \eqref{eq:derN} that $D(r)$ is strictly increasing, when $r$ is big enough. From this it follows that $E(r)$ is strictly increasing for $r$ large enough.
Indeed, assume that $s,r$ are large enough with $r<s$ and  
	\begin{equation}\label{assumptionE}
		E(\di)\geq E(s).
	\end{equation}
	From this assumption it follows, with the help of the second resolvent formula and \eqref{eq:IMScor2}, that
	\begin{equation}
		-(\Hx_{\di}^\bot-E(s))^{-1} \geq -(\Hx_{\di}^\bot-E(\di))^{-1}. \label{ineq:Resolvent}
	\end{equation}
	Using \eqref{ineq:Resolvent} we observe that
	\begin{align}\nonumber
		D(\di) &= \langle\Psi_{\di}, \Hx_{\di}\Psi_{\di}\rangle - \langle\Psi_{\di}, \Hx_{\di}\Pi_{r}^\bot(\Hx_{\di}^\bot-E(s))^{-1}\Pi_{r}^\bot\Hx_{\di}\Psi_{\di}\rangle
		\\ \label{eq:mon} &\geq  \langle\Psi_{\di}, \Hx_{\di}\Psi_{\di}\rangle - \langle\Psi_{\di}, \Hx_{\di}\Pi_{r}^\bot(\Hx_{\di}^\bot-E(\di))^{-1}\Pi_{r}^\bot\Hx_{\di}\Psi_{\di}\rangle \geq E(r),
	\end{align}
	where the last inequality follows from the fact that  $E(r)$ is the minimum of the spectrum of the Feshbach map $F_{\Pi_\di}(E(r))$. 
	Since moreover $E(s) = D(s)$ it follows from \eqref{assumptionE} and \eqref{eq:mon} that 
	\begin{equation*}
		D(\di) \geq  D(s),
	\end{equation*}
	contradicting the fact that $D$ is strictly increasing for large $r$. 
\end{proof}


\begin{thebibliography}{DGMS}
	
	\bibitem[An]{A} I. Anapolitanos: Remainder estimates for the long range behavior of the van der Waals interaction energy.
	\textit{Annales Henri Poincar\'e}, \textbf{17}(5), 1209-1261 (2016).
	
	\bibitem[AnSi]{AS} I. Anapolitanos, I.M. Sigal: Long Range behavior of van der Waals force.
	\textit{Comm. Pure Appl. Math.}, \textbf{70}, 1633-1671 (2017).
	
	\bibitem[AnLe]{AL} I. Anapolitanos, M. Lewin: Compactness of molecular reaction paths in quantum mechanics. ArXiv:1809.06110 (2018).
	
	\bibitem[BaFrSi]{BFS} V. Bach, J. Fr\"ohlich and I.M. Sigal:
	Quantum Electrodynamics of Confined Nonrelativistic Particles. {\it Adv. in Math.} {\bf 137}, 299-395  (1998).
	
	
	\bibitem[CaSc]{CaS} \'E. Canc\`es, L.R. Scott: Van der Waals Interactions Between Two Hydrogen Atoms: The Slater--Kirkwood Method Revisited.
	{\it SIAM J. Math. Anal.}, {\bf 50}(1), 381--410 (2018). 
	
	\bibitem[ChChJoRiYu]{C} S.J. Cha, Y.G. Choe, U.G. Jong, G.C. Ri, C.J. Yu: Refined phase coexistence line between graphite and diamond from  density-functional Theory and van der Waals correction. {\it Physica B: Condensed Matter}
	{\bf 434} 185-193 (2014).
	
	\bibitem[CoSe]{CoSe} J. M. Combes, R. Seiler: Regularity and asymptotic properties of the discrete spectrum of electronic hamiltonians.
	\textit{Int. J. Quantum Chem.}, \textbf{14}, 213-229 (1978).
	
	\bibitem[CoTh]{CT} J.M. Combes, L. Thomas: Asymptotic behavior of eigenfunctions for multiparticle Schr\"odinger
	operators. {\it Commun. Math. Phys.} {\bf 34}, 251-270 (1973).
	
	\bibitem[CyFrKiSi]{CFKS} H.L. Cycon, R.G. Froese , W. Kirsch and B. Simon:
	Schr\"odinger Operators with application to quantum mechanics and
	global geometry. {\it Texts and Monographs in Physics. Springer
		study edition. Springer-Verlag Berlin} (1987).
	
	\bibitem[FeSu]{FS} G. Feinberg and J. Sucher: General theory of the van
	der Waals interaction: a model independent approach. {\it Phys. Rev.
		A} {\bf 9}, 2395-2415  (1970).
	
 	\bibitem[FrGrRiSe]{FGRS}  J. Fr\"ohlich, G.M. Graf, J.-M. Richard and M. Seifert: Proof of stability of the hydrogen molecule. {\it Phys. Rev. Lett.}, {\bf 71}, No.9, 30  1332-1334 (1993).
	
	\bibitem[GrGrHaSi]{GrGrHaSi} S. Graffi, V. Grecchi, E. Harrell, H. Silverstone: The $1/R$ expansion for H$_2^+$: Analyticity, summability, and asymptotics 
	\textit{Annals of Physics}, \textbf{165}, 441 - 483 (1985).
	
	\bibitem[Hu1]{H} W. Hunziker: On the spectra of Schr\"odinger Multiparticle Hamiltonians.
	{\it Helv. Phys. Acta} {\bf 39}, 451-462 (1966).

    \bibitem[Hu2]{H1}  W. Hunziker: Distortion analyticity and molecular resonance curves.
    {\it Ann. Inst. Henri Poincare} {\bf 45}, 339-358 (1986).
	
	\bibitem[HuSi]{HS} W. Hunziker and I.M. Sigal: The quantum $N-$body
	problem. {\it J. Math. Phys.} {\bf 41} No.6, 3448-3510 (2000).
	
	\bibitem[Jo1]{J1} J.E. Jones: On the Determination of Molecular Fields. I.
	From the Variation of the Viscosity of a Gas with Temperature. {\it Proc. R. Soc. Lond. A} {\bf 106}(738), 441-462 (1924).
	
	\bibitem[Jo2]{J2} J.E. Jones: On the Determination of Molecular Fields. II.
	From the Equation of State of a Gas. {\it Proc. R. Soc. Lond. A} {\bf 106}(738), 463-477 (1924).
	
	\bibitem[Le1]{Le} M. Lewin: A Mountain Pass for reacting molecules. {\it Ann. Henri Poincar\'e} {\bf 5} 477-521 (2004).
	
	
	\bibitem[Le2]{Lewin-11} M. Lewin: Geometric methods for nonlinear many-body quantum systems. {\it J. Funct. Anal.} {\bf 260} 3535-3595 (2011).
	
	
	\bibitem[LiLo]{LL} E.H. Lieb and M. Loss: Analysis. {\it Graduate Studies in Mathematics} {\bf 14} AMS, Providence, RI, second edition (2001).
	
	
	\bibitem[LiTh]{LT}  E.H. Lieb and W. Thirring : Universal nature of van der Waals forces for Coulomb
	systems. {\it Phys. Rev. A} {\bf 34} No.1, 40-46 (1986).
	
	\bibitem[Lo]{Lo} F. London: The general theory of molecular forces.
	\textit{Transactions of the Faraday Society}, {\bf 33}, 8-26 (1937).
	
	\bibitem[MiSp]{MS} T. Miyao, H. Spohn: The retarded van der Waals potential: Revisited. J. Math. Phys. {\bf 50}, 072103 (2009).
	
	\bibitem[MoSi]{MoS} J.D. Morgan, B. Simon:
	Behaviour of Molecular Potential Energy Curves for large Nuclear Separations. {\it Int. J. Quantum Chem.}  XVII 1143-1166 (1980) .
	
	\bibitem[Nam]{Nam} P.T. Nam: New bounds on the maximum ionization of atoms. {\it Commun. Math. Phys.}
	{\bf 312}, 427-445 (2012).
	
	\bibitem[ReSi1]{RSI} M. Reed and B. Simon: Methods of modern Mathematical Physics I: Functional Analysis. {\it Academic Press Inc.} (1980).
	
	\bibitem[ReSi4]{RSIV} M. Reed and B. Simon: Methods of modern Mathematical Physics IV:
	Analysis of Operators. {\it Academic Press Inc.} (1980).
	
	
	\bibitem[Si1]{Sig} I.M. Sigal: Geometric methods in the quantum many-body problem.
	Non-existence of very negative ions. {\it Comm. Math. Phys.} {\bf 85}, 309-324 (1982).
	
	\bibitem[Si2]{Sig2} I.M. Sigal: How many electrons can a nucleus bind? {\it Ann. Phys.} {\bf 157} No.2, 307-320 (1984).
	
	
	\bibitem[vdWa1]{vdW1} J.D. van der Waals: On the continuity of the Gaseous amd Liquid states. Edited and with
	an introduction by J.S. Rowlison. {\it Dover Phoenix Editions.} (1988).
	
	\bibitem[vdWa2]{vdW2} J.D. van der Waals: On the continuity of the Gaseous amd Liquid states. Nobel lecture (1910).
	
	\bibitem[vWi]{vW} C. Van Winter: Theory of Finite systems of Particles. I. The Green function.
	{\it Mat.-Fys. Skr. Dankse Vid. Selsk.} { \bf 2} No.8 (1964).
	
	\bibitem[Zh]{Zh} G.M. Zhilin: Discussion of the Spectrum of Schr\"odinger operators for systems of many particles (In Russian).
	{\it Trudy Moskovskogo matematiceskogo obscestva}, {\bf 9}, 81-120 (1960).
	
	
\end{thebibliography}
\end{document}